\documentclass[11pt]{article}
\usepackage{amsmath,amsthm,amsfonts,amscd,latexsym,amssymb,mathrsfs,dsfont,cancel,mathtools,physics,tensor,comment,appendix,enumerate}
\usepackage{tcolorbox}
\usepackage{graphicx}
\usepackage{hyperref}
        \hypersetup{
           breaklinks=true,   
           colorlinks=false,   
        }

\usepackage[autostyle, english = american]{csquotes}
\MakeOuterQuote{"}
\DeclareMathOperator\WF{WF}
\DeclareMathOperator\supp{supp}

\DeclareMathOperator\sign{sign}

\def\Chi{\text{\normalfont X}}
\def\RC{r}
\def\RQ{R}
\def\q{{q}}
\newtheorem{theorem}{Theorem}[section]

\newtheorem{proposition}[theorem]{Proposition}
\newtheorem{lemma}[theorem]{Lemma}
\theoremstyle{definition}

\newtheorem{remark}[theorem]{Remark}

\usepackage[backend=bibtex,style=alphabetic,sorting=nyt,maxcitenames=5,mincitenames=3,maxbibnames=10,giveninits=true]{biblatex}

\begin{document}
	
	\par 
	\bigskip 
	\LARGE 
	\noindent 
	{\bf An algebraic QFT approach to the Wetterich equation on Lorentzian manifolds} 
	\bigskip 
	\par 
	\rm 
	\normalsize


\large
\noindent 
{\bf Edoardo D'Angelo$^{1,2,a}$}, {\bf Nicol\`o Drago$^{3,b}$}, {\bf Nicola Pinamonti$^{1,2,c}$}, {\bf Kasia Rejzner$^{4,d}$}  \\
\par
\small
\noindent$^1$ Dipartimento di Matematica, Universit\`a di Genova, Italy. 
\smallskip

\noindent$^2$ Istituto Nazionale di Fisica Nucleare - Sezione di Genova, Italy. 
\smallskip
 
\noindent$^3$ Dipartimento di Matematica, Universit\`a di Trento, Italy. 
\smallskip

\noindent$^4$ Department of Mathematics, University of York, UK.
\smallskip

\noindent E-mail: 
$^a$edoardo.dangelo@edu.unige.it, 
$^b$nicolo.drago@unitn.it, 
$^c$pinamont@dima.unige.it, \\
$^d$kasia.rejzner@york.ac.uk
\normalsize

\par 

\rm\normalsize 

\rm\normalsize 


\par 
\bigskip 

\rm\normalsize 
\noindent {\small Version of \today}

\par 
\bigskip

\rm\normalsize

\bigskip

\noindent 
\small 
{\bf Abstract}
%
%

\noindent
We discuss the scaling of the effective action for the  interacting scalar quantum field theory on generic spacetimes with Lorentzian signature and in a generic state (including vacuum and thermal states, if they exist).
This is done constructing a flow equation, which is very close to the renown Wetterich equation, by means of techniques recently developed in the realm of perturbative Algebraic Quantum Field theory (pAQFT).
The key ingredient that allows one to obtain an equation which is meaningful on generic Lorentzian backgrounds is the use of a local regulator, which keeps the theory covariant. 
As a proof of concept, the developed methods are used to show that non-trivial fixed points arise in quantum field theories in a thermal state and in the case of quantum fields in the Bunch-Davies state on the de Sitter spacetime. 

\rm\normalsize 

\tableofcontents

\section{Introduction}

In this paper we extend functional renormalization methods to obtain flow equations for the effective action \textit{à la} Wetterich in curved Lorentzian backgrounds and in generic states.
Here we study the case of a real quantum scalar field, but the main ideas and methods we use here are applicable to other types of fields. 
To derive the flow equations, 
we work in the recently developed framework of perturbative algebraic quantum field theory \cite{HollandsWald2001a,HollandsWald2001b, HollandsWald14, Fredenhagen2012, AAQFT15, Rejzner2016} (pAQFT) in the functional approach \cite{Brunetti2009}. However, using ideas similar to those of \cite{Buchholz2019} the obtained result is valid also in an exact sense. 

Our method is based on the average effective action approach to the functional renormalization group, which was developed in a broad range of applications from the seminal papers of Wetterich \cite{Wetterich1989, Ringwald1989, Wetterich1991}. In Wilson's modern formulation of the renormalization group \cite{Wilson73}, short-distance fluctuations are progressively integrated-out, obtaining a coarse-grained, low-energy description of a system from its microscopic degrees of freedom.

In the functional approach \cite{Berges2000}, the fundamental object is the generating functional for the 1PI Feynman diagrams, regularised introducing a momentum cut-off which suppresses long-range fluctuations. The average effective action acts as a microscope with variable resolution, which permits to move from the fine-grained, microscopic description to the rough, macroscopic view. The equation governing the flow of the fluctuations from the microscopic to the macroscopic scale is the Wetterich equation \cite{Wetterich1992}, developed from earlier ideas of Polchinski \cite{Polchinski1983}.
Within the functional approach, non-perturbative methods to approximate the Wetterich equation have been developed, opening the possibility to study non-perturbative effects in highly-correlated systems.
The Functional Renormalization Group (FRG) has found numerous applications ranging from statistical physics to high-energy particle physics (QCD in particular), with the generalization of the Wetterich equation to include gauge fields \cite{Reuter1993}, and to approaches to quantum gravity based on the asymptotic safety scenario \cite{Niedermaier2006, Reuter1996, Bonanno2020}.

Most results within the FRG approach have been derived in Euclidean spaces. Investigations on Lorentzian signature FRG flows, based on analytic continuation of Euclidean correlation functions, has been initiated in \cite{Floerchinger2011}. A different approach, based on real-time Schwinger-Keldysh formalism and the spectral representation of correlation functions in Minkowski spacetime, has been developed in \cite{Kamikado2013, Pawlowski2015, Horak2020, Horak2021, Berges2012, Huelsmann2020} and is currently under investigation. Finally, a Lorentzian study of Asymptotic Safety in quantum gravity, based on a 3+1 decomposition of the metric, can be found in \cite{Manrique2011}.

Nice reviews on the use of the functional renormalization group method related to asymptotic safety can be found in \cite{Eichhorn2019,PawlowskiReichert2021,Pereira2019,ReuterSaueressig2012,ReuterSaueressig2019}.
Recent important developments in the latter research area include \cite{Fehre2021}, which postulates a flow equation for the graviton spectral function in Lorentzian signature.
Our current work adds to this by providing a framework where this flow equation can be derived from first principles and further generalized to curved spacetimes with Lorentzian signature.

The Wetterich flow equation is usually derived in the following way:
one starts with considering $I+J$, where $I$ is the action of the given theory and $J=\int j\chi$ is the term describing the smearing of the field $\chi$ with an external current $j$.
Consider $Z(j) = \langle \exp(i I + i J )\rangle$, the generating functional for the time-ordered correlation functions of the theory. The connected, time-ordered correlation functions are obtained from $W= -i\log Z$ by means of functional derivatives. 
Having a full control over $Z$ or $W$ would allow one to obtain the precise form of the correlation functions of the theory and hence would provide access to the interacting quantum theory.

Since the direct analysis of $Z$ is usually impractical, it is slightly better to study the associated effective action $\tilde\Gamma$.  
To obtain $\tilde\Gamma$, one starts by introducing the field $\phi$ defined as the first functional derivative of $W$, and thus depending on $j$. In the literature $\phi$ is called the {\it  classical field} because it arises from an expectation value.
The relation between $\phi$ and $j$ can be inverted, at least in perturbation theory, giving $j=j_\phi$.   
The effective action $\tilde\Gamma(\phi)= W(j_\phi)- \int \dd^d x j_\phi \phi $ is then obtained from $W$ by means of a Legendre transform.

The effective action $\tilde \Gamma$ encodes all the information on the quantum correlation functions, as the original $Z(j)$. Even though the effective action $\tilde \Gamma$ is local in the free case, when one includes interactions it contains an infinite series of possibly non-local terms, and as such can be computed only perturbatively, in some cases as a sum of the 1PI Feynman diagrams. As such, one introduces a suitable regularization, taming the infrared (long-range) quantum fluctuations. In order to do so, one artificially adds a contribution $Q_k$, quadratic in the fields, to the action in the generating functional $Z(j)$. The regulator $Q_k$ depends on a scale parameter $k$ and hence both the generating functional $Z_k$ or $W_k$ and the effective action $\tilde\Gamma_k$ depend on $k$. Their behaviour under changes of the parameter $k$ is governed by certain flow equations.

In particular, the Polchinski equation is an equation for the derivative with respect to $k$ of $W_k$, while the Wetterich equation involves the $k$-derivative of $\Gamma_k(\phi)=\tilde\Gamma_k(\phi)-Q_k(\phi)$ and it takes the well known form
\[
\partial_k \Gamma_{k} = \frac{i}{2}
\big\langle ( \Gamma^{(2)}_k + Q_{k}^{(2)} )^{-1}, \partial_k Q_{k}^{(2)}\big\rangle_{2}\,,
\]
where $\langle\cdot,\cdot\rangle_2$ is the standard pairing on $\mathcal{M}^2=\mathcal{M}\times\mathcal{M}$. Furthermore, $( \Gamma^{(2)}_k + Q_{k}^{(2)} )^{-1}$ is just $W_{k}^{(2)}$. Hence, the non-uniqueness of the inverse of $(\Gamma^{(2)}_k + Q_{k}^{(2)})$, present in the Lorentzian case, is not an issue when $W$ is the starting point of the construction, so that $W_{k}^{(2)}$ is the distinguished inverse. 

In order for this equation to be useful, the regulator $Q_k$ needs to have certain properties
\cite{Litim2006, Litim1998}:
\begin{itemize}
\item it should vanish in the limit $k\to0$, so that the original theory is recovered in that limit;
\item it should suppress all the quantum fluctuations in the limit $k\to\infty$, so that in that limit one obtains a theory governed by a classical action;
\item at finite $k$, it should behave as an effective mass term to control potential infrared divergences;
\item at finite $k$, it should vanish at high momentum to not alter drastically the short distance behaviour of the correlation functions.  
\end{itemize}

In the original approach and for Euclidean field theories, $Q_k$ is chosen as a momentum cutoff. One of the most used sharp cut-offs assumes a simple expression in the Fourier transform of its second functional derivative, as $\hat{Q}_k^{(2)}(p) = -(k^2-p^2)\theta(k^2-|p|^2)$, where $\theta$ is the Heaviside step function.  It has been discussed in \cite{Litim2001} that this regulator meets all the requirements listed above, and furthermore permits to keep the technical difficulties in practical computations under control.
When such a regulator is used, the source term at the right-hand side of the flow equation of the effective action $\tilde\Gamma_k$ has a peak in a vicinity of $|{p}|^2\sim k^2$, while both high and low momentum modes are suppressed. This gives rise to a flow in Wilsonian sense, for which at scale $k$ only the spectrum of the various propagators at momentum squared equal to $k^2$ matters, thus providing an interpretation of the used regularization at scale $k$ as a coarse-graining procedure.
Unfortunately, in position space, the regulator $Q_k$ introduced above turns out to be non-local. 
For this reason, it is difficult to extend similar techniques to field theories on generic curved backgrounds.
Similarly, if the state in which the theory is constructed is not a vacuum, it is not clear if this choice of regulator completely regularises the theory. This happens, for example, with the Wetterich equation in the case of thermal fields \cite{Tetradis1992, Litim1998}. Another problem of a non-local regulator is that the original analysis is conducted assuming the vacuum state and for Euclidean quantum field theories, and the na\"ive translation of $\hat{Q}_k^{(2)}(p)$ to spacetimes with Lorentzian signature would alter the principal symbol of the equation of motion governing the evolution. 
This could potentially affect the results on the solvability of the equation of motion of the theory even in the non-interacting case. 

As discussed in \cite{Litim2006}, it is still possible to use a local, mass term regulator $Q_k^{(2)} \sim k^a$ at the price of introducing a different regularization procedure of the ultraviolet regime.
This last requirement is not an issue in approaches to interacting field theories, which are automatically ultraviolet finite, like those analysed in \cite{Brunetti2009}. 
When a local regulator is used, the source term of the flow equation for the effective action is not anymore peaked around momenta of scale $k$ and hence the interpretation of the flow equation one obtains gets modified. The flow equations can then be interpreted as the flow of the theory under variations in the mass parameter. 

Furthermore, if $Q_k$ is local, the perturbatively constructed $S$-matrix (used to build interacting fields needed to describe the generating functional $Z(j)$) is formally unitary. 
This implies in particular that, in the Lorentzian case, the effective action obtained from that $W(j)$ is real-valued. On the contrary, the $S$-matrix constructed with non local regulators is in general non-unitary (for states which are not the Minkowski vacuum, see for example \cite{Weinberg1987}) and thus the corresponding effective action could in principle be complex-valued, with an imaginary contribution due to the form of the non-local regulator, and not to intrinsic properties of the investigated physical system. We refer to \cite{Donoghue2020} for a discussion on the issues arising in the connection between Euclidean and Lorentzian approaches to the Wetterich equation.

Finally, a mass term regulator appears to be useful whenever one is interested in preserving the analytical structure of the propagator, without introducing cuts or poles. Such a propagator corresponds to the \textit{Callan-Symanzik propagator} with $Z_k = 1$ defined in Eq. (5) in \cite{Fehre2021}, where the FRG is applied to the spectral function of the graviton propagator.

In this paper, we introduce a version of the effective action that can be analysed on generic curved spacetimes and in generic quantum states. In order to achieve this result, we employ methods of the perturbative Algebraic Quantum Field Theory (see e.g. \cite{AAQFT15, Fredenhagen2012, HollandsWald14, Rejzner2016} for reviews and the references therein). 
In particular, the formulation of interacting theories provided by pAQFT does not make use of any representation provided by the reference state, and in the renormalization procedure used to perturbatively construct interacting fields only the position-space representation of the propagators is needed, with no reference to their Fourier transform.
Every element of the theory is by construction ultraviolet regular; this is particularly useful for our purposes, because there is no need to select a $Q_k$ which cuts the high momenta. 
Following the ideas similar to those presented in \cite{Litim1998} and to obtain flow equations which are tractable on generic spacetime, we choose the regulator to be local in the field to keep covariance. Moreover, the local regulator $Q_k$ acts as an artificial mass contribution to the field to tame infrared problems.
Its form is 
\[
Q_k(\chi) = -\frac{1}{2}\int \dd^d x \q_k (x) \chi^2(x) \ ,
\]
where $\q_k$ grows as $k^2$ for bosonic fields (the only ones we consider here in order to keep the discussion simple). 
Local regulators like the one analysed here have been already used in the literature, in \cite{Fehre2021} as discussed above but also e.g. in \cite{Litim1998, Litim2006}, to analyse states at finite temperature.

\bigskip

In this paper we show that an equation similar to the Wetterich equation can be obtained for the effective action on generic spacetimes and for generic background states. Furthermore, we show that in the case of the Minkowski vacuum it gives results qualitatively analogous to the one obtained with non local regulators.
However, the presence of a local $Q_k$ ---whose second functional derivative contains a Dirac delta function--- leads to two main differences:
(a) the pairing on the right-hand side of the Wetterich equation effectively acts only on $\mathcal{M}$ and not on $\mathcal{M}^2$;
(b) the contribution of $ ( \Gamma^{(2)}_k + Q_{k}^{(2)} )^{-1}(x,y)$ needs to be evaluated at coinciding points.
This second modification is not an issue because the local fields present in $Q_k$ are normal-ordered, and so the coinciding point limit we have to consider is finite without requiring extra regularizations.
Furthermore, a key difference from the Euclidean case is that the inverse $ ( \Gamma^{(2)}_k + Q_{k}^{(2)} )^{-1}$ is not unique
in Lorentzian spacetimes. The choice of the inverse we make is $( \Gamma^{(2)}_k + Q_{k}^{(2)} )^{-1} = W_{k}^{(2)}$ and it 
 depends on the choice of a reference state, and introduces a state dependence in the flow equations.  
 
Taking into account all this, we conclude that the modified Wetterich equation takes the form
\[
\partial_k \Gamma_{k} = \lim_{y \to x} -\frac{i}{2} \int \dd x \partial_k \q_{k}(x) \left [ \left( \Gamma^{(2)}_k - \q_{k}\right)^{-1}(x,y) - \widetilde{H}_F (x,y) \right ] \ .
\]
On the right-hand side $\widetilde{H}_F(x,y)$ is a counter-term implicitly defined ---\textit{cf.} Equation \eqref{Eq: implicit definition of tilde HF}--- and it is related to a Hadamard parametrix whose asymptotic behaviour in the limit $x\to y$ is universal, and can be obtained from just the background geometry and the free (linearised) equation of motion of the theory  \cite{Brunetti2001,HollandsWald2001a,AAQFT15}.
The subtraction described above appears to be very close to the known point-splitting regularization usually employed to get expectation values of Wick powers in curved backgrounds.  
The use of a Hadamard parametrix constructed with local properties of the metric only, and not the subtraction of the two-point function of a state, is necessary to keep the theory covariant \cite{HollandsWald2001a, Brunetti2001}.
As we shall see later, its presence is essential in the case of Minkowski vacuum to get 
a flow of the effective action qualitatively similar to the one obtained with non local regulator.

We also discuss an approximation scheme which is used to get approximated solutions of the obtained flow equation.
The developed method is then applied in special cases. The first case is a quantum field theory on
a Minkowski spacetime in the vacuum, in which our results are similar to the one obtained with non-local regulators. We then pass to discuss the case of fields in a thermal state, getting results similar to those presented in \cite{Litim2006}.
Finally, to prove that the method is directly applicable to the curved background, we analyse a quantum field theory in the Bunch-Davies state on the de Sitter spacetime.
We refer to \cite{banerjee2022spatial} for an analysis, similar in spirit to the one presented in this paper, focused on cosmological spacetimes and on the case of a non-local regulator.

The paper is organized as follows. In Section \ref{Sec: Perturbative Algebraic Quantum Field Theory} we collect some known facts about pAQFT methods used in the paper. Section 3 contains the analysis of the effective action and its regularization. In Section 4 we present the derivation of the flow equation. Section 5 contains the analysis of the approximation methods used, and finally, in section 6 we discuss some explicit examples of the flow equations. Conclusions and outlook are presented in Section 7. We collect some more technical results in the appendices.

\section{Perturbative Algebraic Quantum Field Theory}\label{Sec: Perturbative Algebraic Quantum Field Theory}
    \subsection{Classical field theory}
    In this section, we give a brief overview of the functional approach used in the perturbative construction of interacting quantum field theory, recently developed in the context of Algebraic Quantum Field Theory (AQFT) \cite{Brunetti2009}. This approach can be applied in full generality to gauge field theories, but, in this paper we restrict our attention to scalar fields for simplicity.
    For a thorough treatment and further references, see \cite{AAQFT15}. In what follows, we consider a $d$-dimensional globally hyperbolic spacetime $(\mathcal M, g)$ \cite{HawkingEllis73} whose metric has signature $-+\ldots+$.
    
    In this approach, classical observables are described by complex-valued functionals $F \in \mathcal F$ over \textbf{off-shell field configurations} $\chi \in C^\infty(\mathcal M, \mathbb R)$ with certain properties. 
    In particular, to implement locality, we require that the functionals have compact support, i.e.,
        \begin{gather*}
    \supp F := \{ x \in \mathcal M \ | \ \forall \ \text{neighbourhoods} \ U \ \text{of} \ x \ \exists \ \chi, \psi \in C^\infty(\mathcal M^n, \mathbb R), \ 
    \\
    \supp \psi \subset U : F(\psi + \chi) \neq F(\psi)  \}
    \,.
	 \end{gather*}
Furthermore, the functionals we are working with are smooth with respect to functional derivatives, in the sense that for every $n$, the $n$-th functional derivative obtained as 
\[
\eval{\dv[n]{t} F(\chi + t\psi)}_{t = 0} = \langle F^{(n)}(\chi), \psi^{\otimes n} \rangle \, \qquad \chi,\psi\in C_{\mathrm{c}}^\infty(\mathcal{M})\,,
\]
is a well defined, compactly supported, symmetric distribution.  
To keep the presentation simple, we also require that elements of $\mathcal{F}$ have only finitely many non-vanishing functional derivatives. 

Finally, we restrict the set to the \textbf{microcausal functionals}, satisfying a particular condition on their wavefront set:
\[
\mathcal F_{\mu c} := \{ F \in \mathcal F \ | \ F \text{\; is smooth, compactly supported, and} \ \WF(F^{(n)}) \cap (\overline{V}^n_+ \cup \overline{V}^n_-) = \emptyset \} \ ,
\]
where $\overline V_{+ (-)}$ denotes the closure of the subset of the cotangent space whose elements have covectors in the future (past) light-cones. 
The vector space $\mathcal F_{\mu c}$
is equipped with a weak topology induced by the natural topologies of distributions. 
Actually, we say that $A_l\in \mathcal F_{\mu c}$ converges to $A\in \mathcal F_{\mu c}$ for $l\to \infty$ if for every $n$ and for every field configuration $\chi\in C^{\infty}_{\mathrm{c}}({\mathcal{M}};\mathbb{R})$, $A^{(n)}_l(\chi)$ converges to $A^{(n)}(\chi)$ in $\mathcal{D}'(\mathcal{M}^{n})$. 
We refer to \cite{Rejzner2016} for further details.
Two important subsets of the space of microcausal functionals are the \textbf{local functionals} $\mathcal F_{\text{loc}}$, whose $n$-th derivatives $F^{(n)}$ are only supported on the diagonal $\mathcal D_n:=\{(x_1,\dots, x_n)\subset \mathcal{M}^n \;|\; \forall i,j\; x_i=x_j\}$, satisfy $\WF(F^{(n)}) \perp T\mathcal D_n $ and are used to describe local interaction Lagrangians, and the \textbf{regular functionals} $\mathcal F_{\text{reg}}$, satisfying $\WF(F^{(n)}) = \emptyset \ \forall n$.

For example, consider the linear fields and the Wick powers used to construct the interaction Lagrangians,
\[
\Chi_f(\chi)
:= \int_{\mathcal{M}} \dd^d x f(x) \chi(x), \qquad 
\Chi^n_f(\chi)
 := \int_{\mathcal{M}} \dd^d x f(x) \chi(x)^n, \qquad f\in{C}^\infty_0(\mathcal{M})\,, 
\]
which are elements of $\mathcal{F}_{\mu c}$ for every $n$. Here and in the rest of the paper, $\dd^dx$ denotes the volume form on $\mathcal{M}$ induced by $g$ and the orientation of $\mathcal{M}$. Furthermore,  
we have that $\Chi_f$ is local and also regular while $\Chi_f^n$ with $n>1$ is local but not regular. 
In the following, we will often use the integral kernels of these local functionals with respect to $f$, and we denote them simply by $\chi^n(x)$.

The space $\mathcal{F}_{\mu c}$ is linear, and equipping it with the pointwise product $F\cdot G (\chi) := F(\chi) G(\chi)$ and with the involution $F^*(\chi) := \overline{F(\chi)}$, we obtain a commutative $*$-algebra denoted by $(\mathcal{F}_{\mu c},\cdot,*)$.
The latter is the off-shell algebra of classical observables of the classical field theory.

    \subsection{Deformation quantization}
In the case of quantum theories satisfying hyperbolic linear equations of motions, the quantum observables algebra is obtained from $(\mathcal{F}_{\mu c},\cdot,*)$ by deforming the pointwise product to a suitable non-commutative, associative product which encodes the canonical commutation relations.
Concretely, we consider a free action for a scalar field
\begin{equation}\label{eq:free-action}
I_0(\chi) = - \int_\mathcal M \dd^d x \left( \frac{1}{2} \nabla_a\chi \nabla^a \chi + \frac{\xi}{2}R \chi^2 + \frac{m^2}{2} \chi^2 \right) f \ ,
\end{equation} 
where $m$ is the mass of the field and $\xi$ its coupling to the scalar curvature $R$.
Furthermore, $f\in C_{\mathrm{c}}^\infty(\mathcal{M},\mathbb{R})$ is an infrared cutoff, which is a positive function equal to one on the portion of spacetime over which we want to test our theory and guarantees that $I_0$ is an element of $\mathcal{F}_{\mu c}$. This cutoff is eventually removed by taking the \textit{adiabatic limit} $f \to 1$ in a suitable way \cite{Brunetti1999}.

From the above action, one derives the equations of motion $P_0 \chi = 0$, where $P_0$ is the linear, hyperbolic differential operator
\[
P_0 = \square - m^2 - \xi R \ .
\]
Here, $\square$ is the d'Alembert operator associated with the metric $g$. On globally hyperbolic spacetimes, such operator admits unique advanced and retarded fundamental solutions (or propagators) $\Delta_{A,R}$ which in turn define the {\bf causal propagator} (or commutator function) $\Delta = \Delta_R - \Delta_A$. The commutator function is then used to deform the commutative pointwise product of elements of $\mathcal{F}_{\text{reg}}$, obtaining the standard quantum product. Concretely, we define the quantum product $\star$ on $\mathcal F_{\text{reg}}$ as
\[
F \star G := \textsf{M} \circ e^{\Upsilon_{i \Delta/2}} (F \otimes G) \ , \quad \Upsilon_{\Delta} := \int_{\mathcal M^2} \Delta(x,y) \frac{\delta}{\delta \chi(x)} \otimes \frac{\delta}{\delta \chi(y)} \dd x \dd y \, \qquad 
F,G\in\mathcal{F}_{\text{reg}},
\]
where $\textsf{M}$ maps the tensor product to the pointwise product, $\textsf{M}(F\otimes G)(\chi)= F(\chi)G(\chi)$. 
More explicitly,
\begin{equation} \label{star-product}
F \star G = FG + \sum_{n\geq 1}^\infty \frac{1}{n!} \langle F^{(n)}, {\bigg(\frac{i}{2}\Delta\bigg)}^{\otimes n} \ G^{(n)} \rangle \ .
\end{equation}
Such a product implements canonical commutation relations between linear fields, in the sense that
\[
[\Chi_f,\Chi_h]_\star = \Chi_f \star \Chi_f - \Chi_f \star \Chi_f = 
i \langle f, \Delta \ h \rangle \ , \qquad f,h\in C_{\mathrm{c}}^\infty(\mathcal{M})
\]
and it is compatible with the involution $*$, $(F \star G)^* = G^* \star F^*$.
Therefore, the off-shell algebra of regular observables is given by
\[
\mathcal A_{\text{reg}} = (\mathcal F_{\text{reg}}, \star, *) \ .
\]
The free algebra $\mathcal A_{\text{reg}}$ is in fact generated by the identity, together with all possible linear fields $\{ \Chi_f \ | \ f \in C_{\mathrm{c}}^\infty(\mathcal M) \}$,
actually, every element of $\mathcal{A_{\text{reg}}}$ can be obtained as the limit of a sequence of linear combinations of products of its generators. The convergence of this sequence is taken with respect the topology of $\mathcal{F}_{\mu c}$ we have briefly recalled above.

However, the algebra $\mathcal A_{\text{reg}}$ is too small to define a quantum theory, since the product written above cannot be directly extended to non-linear local functionals, like those necessary to describe interaction Lagrangians, or even the stress-energy tensor of a free theory, because they are too singular. 
In order to obtain a well defined product among generic local fields we have to further deform the product. 
This is done by using in the construction of the $\star$-product a suitable bidistribution $\Delta_+$ in place of $i\Delta/2$, of the form
\begin{equation}\label{eq:two-point}
\Delta_+ := \Delta_S + \frac{i}{2} \Delta \ ,
\end{equation}
where $\Delta_S$ is a real and symmetric distribution, while  $\Delta_+$ solves the linear equation of motion $P_0$ in the weak sense, and its wave front set satisfies the \textbf{microlocal spectrum condition} \cite{BFK, Radzikowski1996} 
\begin{equation}\label{eq:microlocal-spectrum-condition}
\WF(\Delta_+) = \{(x,y;k_x,k_y)\in T^{*}(\mathcal{M}^2)\setminus \{0\} \ |\ (x,k_x)\sim(y,-k_y), k_x \triangleright 0 \}
\end{equation}
where $(x,k_x)\sim(y,-k_y)$ holds if $x$ and $y$ are joined by a null geodesic $\gamma$, $g^{-1}k_x$ is tangent to $\gamma$ at $x$ and $-k_y$ is the parallel transport of $k_x$ along $\gamma$. Furthermore, $k_x\triangleright 0$ holds if $g^{-1} k_x$ is future pointing. It is known that states that are quasifree and have a two-point function $\omega_2$ satisfying 
this condition exist. Furthermore, \cite{Radzikowski1996} has shown that these two-point functions have an universal singular structure typical of Hadamard parametrices \cite{KayWald}. We shall be more precise on that in the next section.

Given $\Delta_+$, $\mathcal F_{\mu c}$ becomes a $*$-algebra with the quantum product $\star_{\Delta_+}$ defined by
\[
F \star_{\Delta_+} G = \textsf{M} \circ e^{\Upsilon_{\Delta_+}} (F \otimes G) \ .
\]
The canonical commutation relations between linear fields hold also with this further deformed product. Moreover, the $*$-subalgebra 
$(\mathcal{F}_{\text{reg}},\star_{\Delta_+},*)$ is isomorphic to $\mathcal{A}_{\text{reg}}$. 
The isomorphism $\alpha:\mathcal{A}_{\text{reg}} \to (\mathcal{F}_{\text{reg}},\star_{\Delta_+},*)$ is realised by 
\begin{equation}\label{eq:alpha}
\alpha_{\Delta_S}(F) = e^{\widetilde{\Upsilon}_{\Delta_S}}F, \qquad  \widetilde{\Upsilon}_{\Delta_S}   
= \frac{1}{2}\int_{\mathcal M^2} \Delta_S(x,y) \frac{\delta^2}{\delta \chi(x)\delta \chi(y)} \dd x \dd y.
\end{equation}
We have thus obtained the extended algebra $(\mathcal F_{\mu c}, \star_{\Delta_+}, *)$,
which contains also local functionals.

The construction we have presented depends on the non-canonical choice of $\Delta_S$ (or equivalently, of $\Delta_+$) in \eqref{eq:two-point}.
However, different choices of $\Delta_+$ produce isomorphic extended algebras, the isomorphism being defined by
\[
\alpha_{\tilde{\Delta}_+-\Delta_+}: (\mathcal F_{\mu c}, \star_{\Delta_+}, *) \to (\mathcal F_{\mu c}, \star_{\tilde{\Delta}_+}, *)
\]
with $\alpha$ given in \eqref{eq:alpha}. Hence, the $*$-algebras obtained with different two-point functions 
satisfying the properties stated above are equivalent realizations of the same \textbf{extended algebra of fields},  which we denote by $\mathcal{A}$.
The algebra $\mathcal{A}$ is thus seen as an abstract $*$-algebra and every 
$(\mathcal F_{\mu c}, \star_{\Delta_+}, *)$
is a concrete faithful representation of $\mathcal{A}$.
Abstract elements of $\mathcal{A}$ are represented in $(\mathcal F_{\mu c}, \star_{\Delta_+}, *)$ by means of $\alpha_{\Delta_+}$.
Hence,
a particular choice of representation $(\mathcal F_{\mu c}, \star_{\Delta_+}, *)$ of $\mathcal{A}$ can be understood as a choice of reference frame to be used to represent observables.
However, this choice cannot play any role in the construction of physically relevant observables. In this respect, we observe that linear fields are invariant under the action of the isomorphisms $\alpha_{\tilde{\Delta}_+-\Delta_+}\Chi_f=\Chi_f$. However, this is not the case for $\Chi_f^n$ $n>1$: for example, 
\[
\alpha_{\tilde{\Delta}_+-\Delta_+}\Chi^2_f =  \Chi^2_f +  \int_\mathcal{M} \dd^d x (\tilde{\Delta}_+-\Delta_+)(x,x) f(x).
\]
Hence $\Chi_f^2$ in $(\mathcal F_{\mu c}, \star_{\Delta_+}, *)$ differs from $\Chi_f^2$ in $(\mathcal F_{\mu c}, \star_{\tilde{\Delta}_+}, *)$, and furthermore both of them cannot be covariant fields.

We shall take this observation into account later, when we discuss the form of Wick-ordered polynomials.
In particular, Wick powers $:\!\chi^n\!:_H\in\mathcal{A}$ normal ordered with respect to the Hadamard function $H$ (see below), are the elements of $\mathcal{A}$ constructed as $\alpha_{-H}\Chi_f^n$.

\subsection{States} \label{sec:states}
In order to extract physical predictions from the algebra of observables $\mathcal{A}$, one needs to map the space of functionals to actual numbers, associating to every element of $\mathcal{A}$ its expectation value. 
This is achieved by introducing a state $\omega$, which is a positive, normalised, linear functional, initially given on $\mathcal{A}_{\text{reg}}$ and then extended to $\mathcal{A}$.

Thanks to linearity, a state $\omega$ is determined once its $n$-point functions (correlation functions)
\[
\omega_n(x_1,\dots, x_n):=\omega(\chi(x_1) \star \ldots \star \chi(x_n))
\]
are given as distributions on compactly supported smooth functions.

The algebra of observables is constructed as functionals over off-shell field configurations, as the linear equations of motion can be implemented at the level of states.
Hence, the positive normalised linear functionals on $\mathcal{A}_{\text{reg}}$ we want to work with need to be compatible with the equations of motions. This is done requiring that 
\[
\text{Ker}{(\omega)} \supset \mathcal{F}_{\text{reg}}\cdot  \Chi_{P_0f}, \qquad f\in C_{\mathrm{c}}^\infty(\mathcal{M})
\]
namely requiring that the $n$-point functions $\omega_n$ are weak solutions of the linear equation of motion in any of their entries.

Among all possible states, we need to select a class of sufficiently regular states, in order to extend them on $\mathcal{A}$ by continuity, hence completely characterising the extended states by the same $n$-point functions. (Further details are given in \cite{Rejzner2016}). 
We can do this requiring that the state satisfies the microlocal spectrum condition \cite{Sanders2010}, i.e.: i) the two-point function is such that \eqref{eq:microlocal-spectrum-condition} holds, ii) the one-point function is smooth, and
iii) the truncated $n$-point functions with $n>2$ are also smooth. 
It was proved in \cite{Radzikowski1996} that states which satisfy the microlocal spectrum condition have a two-point function with an universal singular structure, known as Hadamard condition \cite{KayWald}. 
In particular, this implies that, for $y$ in a normal neighbourhood of $x$, the integral kernel of the two-point function $\omega_2$ has the Hadamard form, namely 
\begin{equation}\label{eq:Hadamard}
\omega_2(x,y) = \lim_{\epsilon\to 0^+}\bigg[\frac{u(x,y)}{\sigma_\epsilon(x,y)} + v(x,y) \log( \frac{\sigma_\epsilon(x,y)}{\mu^2})\bigg] + w(x,y) = H(x,y) + \frac{i}{2}\Delta(x,y)+ w(x,y) \ 
\end{equation}
where 
$u$, $v$, and $w$ are smooth functions, $\sigma_{\epsilon}(x,y) = \sigma(x,y)+i\epsilon (t(x)-t(y))$ with $t$ a generic time function, and $\sigma$ is the Synge world function, which is one half of the squared geodesic distance taken with sign. 
The function $u$ is the square root of the van-Vleck-Morette determinant \cite{Poisson2011}, so it is a purely geometric object; $v$ is uniquely fixed by geometry, the coupling constants and mass parameters of the theory, and can be expanded in a formal power series of $\sigma$:
\[
v(x,y) = \sum_{n\geq 0} v_n(x,y) \sigma^n(x,y)\,,
\]
such that only $v_0$ is relevant in the coincidence limit without derivatives. Finally, $w$ remains an arbitrary, smooth function, containing the residual freedom in the choice of the state. The additional freedom in the constant $\mu$ is required to have a dimensionless argument in the logarithm.
Hence, in the coincidence limit, the divergent part of the 2-point function is encoded in the {\bf Hadamard function} $H(x,y)$ and in the causal propagator $\Delta$, which are known a priori.
As examples, it is known that the Minkowski vacuum, or generic thermal states for the free theory in flat spacetime are Hadamard states, as well as the Bunch-Davies states for linear fields on De Sitter spacetime.


Although the main results of this paper hold for an arbitrary state, we will occasionally restrict our attention to 
states that are \textbf{quasifree} or \textbf{Gaussian} for the free theory.
Quasifree states are defined requiring that i) all odd $n$-point functions vanish, and ii) even $n$-point functions can be computed from the two-point function according to Wick's rule \cite{AAQFT15}. Therefore, fixing the symmetric part $\Delta_S$ of the two-point function uniquely identifies a quasifree state. These are the states whose GNS representations are of Fock type. Notice, however, that states that are quasifree for the free theory are not quasifree for the interacting theory.
We finally observe that if we are interested in computing expectation values in a state $\omega$ of $\mathcal{A}$ whose two-point function is $\Delta_+$, it is particularly useful to select the represent $\mathcal{A}$ with 
$(\mathcal{F}_{\mu c},\star_{\Delta_+},*)$ where the $\star$-product is constructed with $\Delta_+$.
Then, the expectation value of 
$F\in \mathcal{A}$ in the state $\omega$, is simply the evaluation of $G=\alpha_{{\Delta_+}}(F)$ on the vanishing configuration, namely $\omega(F) =  \alpha_{{\Delta_+}}(F)(0)= G(0)$.

\subsection{Normal ordering}\label{se:normal-ordering}
Consider a representation $(\mathcal{F}_{\mu c},\star_{\Delta_+},*)$ of $\mathcal{A}$. 
We observe that the deformed $\star_{\Delta_+}$-product on $\mathcal{F}_{\mu c}$
 implements Wick theorem for the product of non-linear observables. 
In fact, by construction 
\[
\chi^2(x) = \lim_{y\to x} \big[\chi(x) \star_{\Delta_+} \chi(y) - \Delta_+(x,y)\big]\,,
\]
where $\chi(x)\star_{\Delta_+}\chi(y)$ is the integral kernel of $\Chi_f\star_{\Delta_+}\Chi_g$ seen as a distribution on $f\otimes g$,
is always finite, and the same holds for higher polynomials. This means that we can understand local functionals like
$\Chi^n_f$ in $(\mathcal{F}_{\mu c},\star_{\Delta_+},*)$ as Wick-ordered monomials of the fields, where the Wick ordering is with respect to $\Delta_+$.
Namely 
\[
:\!\Chi^n_f\!:_{\Delta_+}
:=
\int_{\mathcal{M}} :\!\chi^n\!:_{\Delta_+} f \dd^d x \in \mathcal{A}
\]
where 
\[
:\!\Chi^n_f\!:_{\Delta_+} = \alpha_{-\Delta_+}(\Chi_f).
\]
In this way $:\!\Chi^n_f\!:_{\Delta_+}\in \mathcal{A}$ is represented in $(\mathcal{F}_{\mu c},\star_{\Delta_+},*)$ as $\alpha_{\Delta_+}:\!\Chi^n_f\!:_{\Delta_+} =\alpha_{\Delta_+}\alpha_{-\Delta_+} \Chi_f
=\Chi_f$.

However, as also discussed above, such normal ordering is not covariant, because $\Delta_+$ is globally defined, and also because of the non canonical choice of the symmetric part in $\Delta_+$. Actually, the quasifree state constructed with $\Delta_+$ would represent a preferred reference state, in contradiction with the requirements of the Equivalence Principle. 
One would like to perform normal ordering with local quantities only; a possibility is to use the Hadamard function $H$ given in \eqref{eq:Hadamard}
to extract the local singularity structure from $\Delta_+$, and define a new normal ordering prescription accordingly: 
\begin{equation}\label{eq:Chi2}
:\!\chi^2\!:_H(x) = \alpha_{-H}\chi(x)^2 
= \lim_{y\to x} \big[\chi(x)\chi(y) - H(x,y)\big] 
\ ,
\end{equation}
which is local and generally covariant. The drawback is that now one needs to pay attention to the correction introduced in representing Wick-ordered polynomials with respect to $H$, in the algebra constructed with the $\star$-product defined by $\Delta_+ = H +i\Delta/2 + w$. For example, in view of \eqref{eq:Chi2} we have 
that $:\!\chi^2\!:_H=\alpha_{-H}\chi^2 \in \mathcal{A}$ is represented in $\mathcal{F}_{\mu c}$ as
\[
\alpha_{\Delta_+}:\!\chi^2\!:_H(x)
= \chi^2(x)+w(x,x)
\]
and furthermore
\begin{equation}\label{eq:example-product}
\begin{aligned}
\alpha_{\Delta_+}\left( :\!\chi^2\!:_H(x) \star :\!\chi^2\!:_H(y) 
\right) =& 
\alpha_{\Delta_+}(:\!\chi^2\!:_H(x))  \star_{\Delta_+} \alpha_{\Delta_+}(:\!\chi^2\!:_H(y)) 
\\
= &
(\chi^2(x) +w(x,x)) (\chi^2(y)+w(y,y)) +
\\
 &
+ 4 \Delta_+ (x,y) \chi(x) \chi(y) + 2 \Delta_+(x,y)^2 \ ,
\end{aligned}
\end{equation}
and so its expectation value in the quasifree state $\omega$ whose two-point function is $\Delta_+$ is
\[
\omega(:\!\chi^2\!:_H(x) \star :\!\chi^2\!:_H(y)) = w(x,x)w(y,y) + 2 \Delta_+^2(x,y) \ .
\]
We finally observe that there is some freedom in the choice of $H$ (as for example the length scale $\mu$ in the logarithmic contribution in $H$); furthermore, we could add a covariantly constructed smooth part to $H$ without breaking general covariance. This freedom has been classified in \cite{HollandsWald2001a, HollandsWald2001b, HollandsWald2004} and at the level of the Wick square it reduces to the choice of two real "regularisation" constants $c_1$ and $c_2$
\[
:\!\chi^2\!:_H   =  
:\!\chi^2\!:_{\tilde{H}} + c_1m^2 +c_2 R\,.
\]

In equation \eqref{eq:example-product}, we kept explicit both the normal ordering prescription and the dependence on $\Delta_+$ in the $\star$-product, to clarify their relationship. In the usual QFT notation, one would leave implicit the $\star$-product, writing explicitly the normal-ordering prescription; in what follows, adopting the more usual notation in the mathematical physics literature, 
we will keep the $\star$-products explicit but, without referring to a particular representation, we drop the subfix $\Delta_+$; at the same time, if not strictly necessary, we keep the covariant normal ordering implicit.

    \subsection{Interacting theories}
We now discuss the perturbative construction of interacting fields. 
The action we are working with contains terms which give rise to non-linear contributions to the equation of motion.
In particular, the action takes the form
    \begin{equation}\label{eq:full-action}
    I(\chi) = I_0 + {\lambda} V = - \int \dd^d x \left( \frac{1}{2} \nabla_a\chi \nabla^a  \chi  + \frac{\xi}{2}R \chi^2 + \frac{m^2}{2} \chi^2 +            \lambda \frac{\chi^n}{n!}\right)  f,
    \end{equation}
where, as before, $f$ is a cutoff introduced to keep $I(\chi) \in \mathcal{F}_{\mu c}$.   
The cutoff $f$ is a smooth compactly supported function which is equal to $1$ on the causal completion of the region where we want to test our theory.
The action $I$ is divided into two parts: $I_0$, which coincides with \eqref{eq:free-action} and gives rise to linear equation of motion, and $V$, the interaction Lagrangian. With the methods discussed above, we have now at disposal the free algebra $\mathcal A$, and
interacting observables are constructed as a formal power series in the coupling constant $\lambda$ with coefficients in the free algebra $\mathcal A$.

The perturbative construction of interacting fields makes use of a new operation, the time-ordered product $T$. We start defining the $T$-product in the subset of regular functionals $\mathcal F_{\text{reg}}$ in a manner similar to \eqref{star-product}:
\[
F \cdot_T G = \mathsf{M} \circ e^{\Upsilon_{\Delta_F}} (F \otimes G) \ ,
\]
that is,
\begin{equation} \label{T-product}
F \cdot_T G = FG + \sum_{n \geq 1}^\infty \frac{1}{n!} \langle F^{(n)}, \Delta_F^{\otimes n} G^{(n)} \rangle \ .
\end{equation}
In the above equations, $\Delta_F$ is a Feynman propagator associated with $\Delta_+$,
\[
\Delta_F = \Delta_+ + i \Delta_A = \Delta_S + \frac{i}{2}(\Delta_R + \Delta_A) \ .
\]
However, even if we are working with normal ordered quantities,  
the $T$-product defined on $\mathcal F_{\text{reg}}$ cannot be extended to $\mathcal F_{\mu c}$. Actually,
multiplying local functionals with overlapping support one
encounters divergences that cannot be treated with methods of microlocal analysis.
Nevertheless, at least among local functions, it is possible to construct time ordered products. 
Therefore, we introduce an axiomatic prescription for the $T$-product on local functionals, seen as a symmetric and multilinear map from multilocal functionals $\mathcal F_{\text{loc}}^{\otimes n}$ to $\mathcal A$, satisfying a set of conditions \cite{Brunetti2009, Brunetti1999, HollandsWald2001a, HollandsWald2001b, HollandsWald2004}. In particular, assuming the causal factorization property,
\[
T(F_1,...,F_n, G_1,...,G_m) = T(F_1,\ldots,F_n) \star T(G_1,...,G_m) \ \text{if} \ J^+(\supp F_i) \cap J^-(\supp G_j) = \emptyset \ ,
\]
and the symmetry of $T$, one gets that $T$ is determined for arguments  with pairwise non-overlapping supports, as in \eqref{T-product}.
To extend $T$ also to local functionals with overlapping supports, one may use a recursive procedure on the number of factors, following the method originally presented by Epstein-Glaser \cite{EpsteinGlaser1973}.
Actually, using locality of the factors, multilinearity and  field independence (discussed later), one reduces the problem of constructing the time-ordered product with $n$ elements as the problem of extending suitable distributions $t_n\in C_{\mathrm{c}}^\infty(\mathcal{M}^n\setminus \mathcal{D}_n$) defined outside the diagonal $\mathcal{D}_n$ to the whole $\mathcal{M}^n$ \cite{Brunetti1999,HollandsWald2001a}.
This can be done keeping fixed the Steinmann \cite{Steinmann} scaling degree of the distribution, up to an ambiguity 
which corresponds to the known renormalization freedom.

The action of $T$ on 
local functionals maps local functionals  to covariant normal ordered ones \cite{HollandsWald2001a}, so that $T(F) = :\; F\; :_H$; (see Section \ref{se:normal-ordering} for the description of the normal ordering we are using) when not strictly necessary we keep the latter operation implicit. 
Furthermore, the map $T$, which is originally defined on multilocal functionals, can also be extended to pointwise products of local fields \cite{FR}.
For more details, see \cite{Brunetti2009, Brunetti1999, HollandsWald2001a, HollandsWald2001b, HollandsWald2004, Fredenhagen2012}.

Having a definition of time-ordering at disposal, one can construct the $S$-matrix of a local interaction as
\[
S(V) := e_{\cdot_T}^{i {\lambda}TV} = 
{Te^{i \lambda V}
=}
\sum_n \frac{i^n {\lambda^n}}{n!} T(\underbrace{V\ldots V}_{n \text{ times}}) \ .
\]
The $S$-matrix is an element of $\mathcal{A}[[\lambda]]$ namely, a formal power series in the coupling constant present in {front of} $V$ with coefficients in $\mathcal A$ satisfying
\begin{enumerate}
\item Causality: $S(A+B+C) = S(A+B)\star S(B)^{-1} \star S(B+C)$ if $J^+(\supp A) \cap J^-(C) = \emptyset$;
\item $S(0) = 1, \quad S^{(1)}(0) = 1$;
\item Field independence: $S(V)^{(1)}=iS(V)\cdot_T {\lambda}TV^{(1)}$.
\end{enumerate}

Using the $S$-matrix, we define the \textbf{relative} $S$-matrix as
\begin{equation}
S_V(F) = S(V)^{-1} \star S(V+F) \ 
\end{equation}
where the inverse is taken with respect to the $\star$-product. Furthermore, we have that for every real local $V$, $S(V)$ is formally unitary, so $S(V)^{-1} = S(V)^*$. 
Finally, interacting fields are represented in the free algebra by means of the \textbf{Bogoliubov map} (also called \textbf{quantum M\o ller map})
\begin{equation}\label{eq:Bogoliubov}
R_V(F) = - \frac{i}{\lambda} \eval{\dv{t} S_V(tT^{-1}F)}_{t = 0} = S(V)^{-1} \star [S(V) \cdot_T F] \ .
\end{equation}
We can interpret $R_V(\chi)$ as the interacting field because $R_V(\chi)$ satisfies weakly the equation of motion, in the sense that
\[
R_V(P_0 \chi) + R_V({\lambda}TV^{(1)}) = P_0 \chi\,,
\]
 where $TV^{(1)}$ is just the first functional derivatives of the normal ordered local potential. Hence, since the free equations of motion are encoded in a generic state $\omega$, we have
\[
\omega \big( R_V(P_0 \chi + {\lambda} TV^{(1)}) \big) = 0 \ .
\]
In the following, when not strictly necessary, we shall not write explicitly the formal parameter $\lambda$ in the formulas and we shall denote the algebra of formal power series simply as  $\mathcal{A}[[V]]$.
We also stress that a sequence in $\mathcal{A}[[V]]$ converges if the coefficients of the formal power series converge in the weak topology of $\mathcal{F}_{\mu c}$ mentioned above. See \cite{Rejzner2016} for further details.

\begin{remark}
The pAQFT formalism in the functional approach closely resembles the usual pQFT formalism preferred in the physics literature.  For example, the time-ordered product in the vacuum state in the algebraic setting can be regarded as the generalization of the (often ill-defined) path integral approach in usual QFT, where $n$-point Green functions are computed from a path integral with Gaussian measure
\[
\langle T\chi(x_1)... \chi (x_n) \rangle = \int \mathcal D \chi e^{i I_0} \chi(x_1)...\chi(x_n) \ ,
\]
where $\langle\ldots \rangle$ is the expectation value in the Minkowski vacuum state.

At the same time, the interacting field $R_V(\chi)$ defined by the Bogoliubov formula is equivalent to the field in the interaction picture,
\[
\Chi_I = R_V(\chi) = S(V)^{-1} \star [S(V)\cdot_T \chi] = T(e^{iV})^{-1} T(e^{iV} \chi) \ ,
\]
where in the last equality we dropped the $\star$-product as it is common in the physics literature.
The main difference from the usual construction is that this formalism does not make use of the vacuum representations of fields. At the same time, local observables like the interaction Lagrangian are normal-ordered in a covariant way. 
These two differences make the formalism directly applicable to fields propagating on curved spacetimes, and more adequate to analyse interacting quantum field theories in generic states. 

There is, however, a price to pay.
When dealing with e.g. the scattering theory in QFT on flat spacetime, one usually takes expectation values in the vacuum state on Minkowski; in this case, the Gell-Mann-Low formula permits to simply 
factorise the $\star$-product present in the Bogoliubov map,
\begin{equation}\label{eq:gellmann-low}
\omega(S(V)^{-1} \star S(V)\cdot_T \chi) = \omega(S(V)^{-1}) \omega( S(V)\cdot_T \chi)= \omega(S(V))^{-1} \omega( S(V)\cdot_T \chi)\,,
\end{equation}
at least when the support of $V$ tends to the entire Minkowski spacetime namely when the adiabatic limit is taken and the cutoff $f$ in $V$ in \eqref{eq:full-action} is removed. Assuming without loss of generality that $f$ is equal to $1$ in the neighbourhood of an origin of $\mathcal{M}$, the adiabatic limit is taken replacing the cutoff with $f(x/n)$ and eventually considering the limit $n\to\infty$ of the various expectation values of interests. A discussion about the validity of \eqref{eq:gellmann-low} for the case of massive field can be found in Section 6.2 of \cite{Duetsch:2000nh} 
making use of estimates given in the appendix of \cite{Duetsch:1996eh}. 
In this case, we have that 
\[
\omega(\Chi_I) =  \frac{\omega(T(e^{iV}\chi))}{\omega(T e^{iV})} \ ,
\]
with analogous formulas for the $n$-point functions; in the perturbative expansion of the right-hand side, only time-ordered products appear.
However, for more general states (e.g. thermal states) or on curved backgrounds, the Gell-Mann-Low formula fails in general, and the $\star$-products play an important role as new, oriented (as the product is non-commutative) internal lines in Feynman diagrams. In this sense, the algebraic approach takes directly into account all these effects.

\end{remark}

\section{Functional renormalization}
    \subsection{Generating functionals}
    In this section, we introduce the generating functional $Z(j)$ of the truncated time-ordered products, and by doing so we generalize the results known on flat Minkowski spacetime for quantum field theories constructed over the Minkowski vacuum to curved spacetimes and to generic states.
    
    Let's start with the review of the standard definitions. On flat Minkowski spacetime, the effective action is introduced as the Legendre transform of the generating functional of the connected interacting time-ordered products.
    The latter is usually defined as follows: let $j\in C_{\mathrm{c}}^\infty(\mathcal M)$ and let $J(\chi):=\Chi_j(\chi)=\int j\chi$.
    Denoting by $\omega_0$ the vacuum state on Minkowski spacetime, the generating functional for the interacting time-ordered products is defined by
    \begin{align}\label{Eq: naive Z-generating functional}
    	\mathcal{Z}(j):=\frac{\omega_0(S(V+J))}{\omega_0(S(V))}\,.
    \end{align}
	It then follows that
	\begin{align}\label{Eq: Z-generating functional defining property}
		\frac{\delta^n}{i^n\delta j(x_1)\cdots\delta j(x_n)}\log \mathcal{Z}(j)\big|_{j=0}
		=(\omega_0^c\circ \RQ_{V})(\chi(x_1)\cdot_T\cdots_T\chi(x_n))\,,
	\end{align}
	where $\omega_0^c$ denotes the connected part of $\omega_0$, defined by
	\begin{align}
	    \label{Eq: connected n-point functions}
	    \omega_0^c(\chi(x_1)\star\cdots\star\chi(x_n))
	    :=\frac{\delta^n}{i^n\delta f(x_1)\cdots\delta f(x_n)}\log \omega_0[\exp_\star(i\Chi_f)]\Big|_{f=0}\,.
	\end{align}
	Similarly the connected time-ordered functions of $\omega_0$ are defined by
	\begin{align}
	    \label{Eq: connected time-ordered n-point functions}
	    \omega_0^c(\chi(x_1)\cdot_T\cdots_T\chi(x_n))
	    :=\frac{\delta^n}{i^n\delta f(x_1)\cdots\delta f(x_n)}\log \omega_0[S(\Chi_f)]\Big|_{f=0}\,.
	\end{align}
	
	Notice that the previous equality uses the Gell-Mann-Low formula,
	\begin{align*}
		\omega_0(\RQ_VA)
		=\omega_0(S(V)^{-1}\star[S(V)\cdot_TA])
		=\frac{\omega_0(S(V)\cdot_T A)}{\omega_0(S(V))}\,.
	\end{align*}
	As already discussed, this formula holds in the adiabatic limit, that is,  in the limit where the cutoff $f$ in $V$ given in \eqref{eq:full-action} tends to $1$ on the whole Minkowski space.
	It reduces the complexity in the actual evaluation of $\omega_0(\RQ_V\chi)$ as it requires to compute only time-ordered products.
	Unfortunately, as discussed above, this formula is not valid for states different from the vacuum one or on general curved backgrounds, and it also fails if one does not take the adiabatic limit.
	For this reasons, the definition of $\mathcal{Z}(j)$ has to be modified.
	Our approach is to provide a definition of $Z(j)$ which fulfils the defining property \eqref{Eq: Z-generating functional defining property} and which reduces to formula \eqref{Eq: naive Z-generating functional} for the case of the vacuum state on Minkowski spacetime.
	With this in mind, we define, for an arbitrary but fixed Hadamard state $\omega$ on $\mathcal{A}$,
	\begin{align}\label{Eq: Z-generating functional}
		Z(j):=\omega(S_{V}(J))
		=\omega[S(V)^{-1}\star S(V+J)]
		=\omega[R_VS(J)]\,,
	\end{align}
	out of which Equation \eqref{Eq: Z-generating functional defining property} is verified by direct inspection.
	For $\omega=\omega_0$ (Minkowski vacuum) and in the adiabatic limit, the definitions given in \eqref{Eq: naive Z-generating functional} and in \eqref{Eq: Z-generating functional} coincide because of Gell-Mann-Low formula.

	
	A remarkable property of $\mathcal{Z}$ defined in \eqref{Eq: naive Z-generating functional} is that Equation \eqref{Eq: Z-generating functional defining property} still makes sense for $j\neq 0$; as a matter of fact
	\begin{align}\label{Eq: j not zero defining property of generating functional}
		\frac{\delta^n}{i^n\delta j(x_1)\cdots\delta j(x_n)}\log \mathcal{Z}(j)
		=(\omega_0^c\circ \RQ_{V+J})(\chi(x_1)\cdot_T\cdots\cdot_T\chi (x_n))\,.
	\end{align}
	This shows that $j$-functional derivatives of $\mathcal{Z}(j)$ are physically meaningful also for $j\neq 0$. 

	Unfortunately, $Z(j)$ defined in \eqref{Eq: Z-generating functional} does not satisfy the property \eqref{Eq: j not zero defining property of generating functional}, as one can see, for example, from the following computation: 
	\begin{align}
		\nonumber
		\frac{\delta}{i\delta j(x)}\log Z(j)
		&=\frac{\omega(S(V)^{-1}\star [S(V+J)\cdot_T\chi(x)])}{\omega(S_{V}(J))}
		\\&=\frac{\omega(S_{V}(J)\star \RQ_{V+J}\chi(x))}{\omega(S_{V}(J))}
		=:\omega_J(\RQ_{V+J}\chi(x))\,,
		\label{Eq: omegaJ definition}
	\end{align}
	where $\omega_J\colon\mathcal{A}[[V]]\to\mathbb{C}[[V]]$ is a well-defined linear functional which, however, fails to be positive.

To justify our definition of $Z$, recall that on regular functionals, one can introduce the interacting star product as
	\[
	F\star_V G\doteq R_V^{-1}(R_V(F)\star R_V(G))\,,
	\]
	For a given state $\omega$ of the free theory, the interacting state is defined by $\omega_V\doteq \omega\circ R_V$, so the correlator of $n$ interacting fields in such state is given by:
	\[
	\omega_V(\chi(x_1)\star_V\dots\star_V \chi(x_n))=\omega(R_V(\chi(x_1))\star\dots \star R_V(\chi(x_n))\,.
	\]
	The time-ordered version of $\star_V$ coincides with $\cdot_T$ (see e.g. \cite{Drago2015}) so the time-ordered correlator of $n$ fields in the interacting theory (interacting Green function) is given by
	\[
		\omega_V(\chi(x_1)\cdot_T\dots\cdot_T \chi(x_n))=	\omega\circ R_V(\chi(x_1)\cdot_T\dots\cdot_T \chi(x_n))\,.
	\]
	On the other hand,
	\begin{equation}
	  \eval{\frac{\delta^n Z}{i^n \delta j(x_1)...\delta j(x_n)}}_{j=0} = \omega \circ R_V \left ( \chi(x_1) \cdot_T ... \cdot_T \chi(x_n) \right) \,,
	\end{equation}
	so, for vanishing sources, the functional derivatives of \eqref{Eq: Z-generating functional} give exactly the interacting expectation value of the time-ordered correlation functions.
	
In this sense, the defining property of $Z$ as the generating functional for the time-ordered correlation functions is satisfied also by our definition, which generalises the usual approach to generic states and possibly curved spacetimes. Moreover, as discussed in \cite{Drago2015,Lindner2013}, if the support of $F$ does not intersect the past of the support of $G$, ($F \gtrsim G$), the expectation value in the interacting state of the time-ordered correlation functions, coincide with the expectation value of the correlation function between interacting observables, since, 
\begin{equation}
   F \gtrsim G \Rightarrow \ R_V(F \cdot_T G) = R_V(F) \star R_V(G) \ .
\end{equation}
Coming back to the property \eqref{Eq: j not zero defining property of generating functional}, we show in the lemma below that it cannot be fulfilled if one departs from the Minkowski vacuum. Hence, it is actually not a sensible condition to require for general states.
	\begin{lemma}\label{Lem: no-go result for strong generating functional}
		If $\omega$ does not fulfil the Gell-Mann-Low formula given in \eqref{eq:gellmann-low}, there is no functional $\zeta(j)$ satisfying Equation \eqref{Eq: j not zero defining property of generating functional}.
	\end{lemma}
	\begin{proof}
	Let $\zeta(j)$ be any generating functional fulfilling \eqref{Eq: j not zero defining property of generating functional} for all $n\in\mathbb{N}$.
	For $n=1$ we have
	 \begin{equation}
	 - i \log \zeta(j)^{(1)}(x) = \omega^c \circ R_{V+J}(\chi(x)) \ .
	 \end{equation}
		By direct inspection we have
		\begin{align*}
			A(x_1,x_2):=\frac{\delta^2}{i^2\delta j(x_1)\delta j(x_2)}\log \zeta(j)
			&=\frac{1}{2}\frac{\delta}{\delta j(x_1)}\omega(\RQ_{V+J}\chi(x_2))
			+x_1\leftrightarrow x_2
			\\&=\frac{1}{2}[\omega(\RQ_{V+J}[\chi(x_1)\cdot_T\chi(x_2)])
			\\&\qquad -\omega(\RQ_{V+J}\chi(x_1)\star \RQ_{V+J}\chi(x_2))]
			+x_1\leftrightarrow x_2
			\\&=\frac{1}{2}\sign(t(x_1)-t(x_2))\omega([\RQ_{V+J}\chi(x_1),\RQ_{V+J}\chi(x_2)]_\star)\,,
		\end{align*}
		where we used the symmetry of the left-hand side in $x_1,x_2$.
		By Equation \eqref{Eq: j not zero defining property of generating functional} for $n=2$ the right-hand side should be equal to
		\begin{align*}
			B(x_1,x_2):=\omega(\RQ_{V+J}[\chi(x_1)\cdot_T\chi(x_2)])
			-\omega(\RQ_{V+J}\chi(x_1))\omega(\RQ_{V+J}\chi(x_2))\,,
		\end{align*}
		which in general is not the case, actually at zeroth order in the perturbation parameter, for quasifree states and for $t(x_1)>t(x_2)$
		\[
			A(x_1,x_2)=\frac{i}{2}\sign(t(x_1)-t(x_2)) \Delta(x_1,x_2)), \qquad 
			B(x_1,x_2)= \Delta_F (x_1,x_2),
		\]
		and $\Delta_S$, the symmetric part of the two-point function, is present in $B$ but not in $A$.
	\end{proof}

	Even though the interpretation given by Equation \eqref{Eq: j not zero defining property of generating functional} is not at our disposal, we can still make sense of the non-positive ``states'' $\omega_J$ defined in Equation \eqref{Eq: omegaJ definition}.
	As a matter of fact, $\omega_J$ is a positive state on an algebra $\mathcal{A}_{\circledast}[[V]]$ which is isomorphic to $\mathcal{A}[[V]]$.
	\begin{proposition}\label{Prop: modified star-product algebra}
		Let $U:=S_{V}(J)/\omega(S_{V}(J))\in\mathcal{F}_{\mu c}[[V]]$.
		Let $\mathcal{A}_{\circledast}$ be the $*$-algebra 
        obtained equipping $\mathcal{F}_{\mu c}[[V]]$ 
        with the product $\circledast$ and the $*$-involution $*_\circledast$
		\begin{align}\label{Eq: modified star-product and involution}
			A\circledast B:=A\star U\star B\,,\qquad
			A^{\ast_\circledast}:=U^*\star A^*\star U^*\,.
		\end{align}
		Then, $\mathcal{A}_{\circledast}$ is a unital $\ast$-algebra and $\omega_J$, defined as per Equation \eqref{Eq: omegaJ definition}, is a state on $\mathcal{A}_{\circledast}$.
		Moreover, the map $\varsigma\colon\mathcal{A}\to\mathcal{A}_\circledast$ defined by $\varsigma(A):=U^*\star A$ is a $*$-isomorphism, and $\varsigma^*\omega_J=\omega$.
	\end{proposition}
	\begin{proof}
		By direct inspection, $\circledast$ is associative with unit given by $\boldsymbol{1}_\circledast:=U^*$ --- notice that $U$ is unitary as $V,J\in\mathcal{F}_{\text{loc}}$.
		Moreover $\circledast$ and $\ast_\circledast$ are compatible, meaning that $(A\circledast B)^{\ast_\circledast}=B^{\ast_\circledast}\circledast A^{\ast_\circledast}$.
		Since $\ast_\circledast$ is an involution, we have that $\mathcal{A}_{\circledast}$ is a unital $*$-algebra.
		
		Now, let $\varsigma\colon\mathcal{A}\to\mathcal{A}_\circledast$ be defined by $\varsigma(A):=U^*\star A$.
		Then $\varsigma$ is linear and invertible, and it holds that
		\begin{align*}
			\varsigma(A)\circledast\varsigma(B)
			=U^*\star A\star \cancel{U\star U^*}\star B
			=\varsigma(A\star B)\,.
		\end{align*} 
		It follows that $\varsigma$ is a $*$-isomorphism.
		Finally
		\begin{align*}
			\varsigma^*\omega_J(A)
			:=\omega_J(\varsigma(A))
			=\omega(A)\,.
		\end{align*}
	\end{proof}
	Proposition \ref{Prop: modified star-product algebra} shows that the $j$-functional derivatives of the generating functional $Z(j)$, given in Equation \eqref{Eq: Z-generating functional}, are still physically meaningful for $j\neq 0$.
	As a matter of fact, such derivatives coincide with the connected time-ordered functions \eqref{Eq: connected time-ordered n-point functions} for the state $\omega_J$ on $\mathcal{A}_\circledast$.
	Notice that, as $\mathcal{A}_\circledast[[V]]\simeq\mathcal{A}[[V]]$, the latter state can be interpreted as a state on $\mathcal{A}[[V]]$ too.

	\subsection{Effective action}
	Starting from the generating functional $Z(j)$ defined in Equation \eqref{Eq: Z-generating functional}, we may introduce the effective action $\tilde\Gamma$ using the standard definition.
	Let $W(j)$ be the functional defined by
	\begin{align}\label{Eq: W-definition}
		Z(j)
		=e^{iW(j)}\,.
	\end{align}
	Notice that, on account of Equation \eqref{Eq: Z-generating functional defining property}, we have
	\begin{align*}
		\frac{\delta W}{\delta j(x)}\bigg|_{j=0}
		= {\frac{1}{Z(0)}} \omega( \RQ_{V}(\chi(x))\,.
	\end{align*}
	The effective action $\tilde\Gamma$ is the functional defined by
	\begin{align}\label{Eq: effective action definition}
		\tilde\Gamma(\phi)
		=W(j_\phi)-J_\phi(\phi)\,,
	\end{align}
	where $j_\phi\in C_{\mathrm{c}}^\infty(M)$ is the current defined by
	\begin{align}\label{Eq: jphi-defining property}
		\frac{\delta W}{\delta j}\bigg|_{j=j_\phi}=\phi\,.
	\end{align}
	Proposition \ref{Prop: existence of jphi} given below shows that Equation \eqref{Eq: jphi-defining property} has a unique perturbative solution, so that Equation \eqref{Eq: effective action definition} really defines a functional.
	
     \subsection{Regularised generating functionals}
    
A direct computation of the generating functional $Z(j)$ and/or of the effective action is usually not feasible.
The main idea of the Wilsonian renormalization group is to progressively take into account high-energy degrees of freedom. This is usually done introducing an artificial scale $k$, such that the modes with energy $E < k$ are suppressed.
As discussed in the introduction, one introduces the scale $k$ so that in the limit $k\to 0$ one recovers the standard definition for the partition function, thus taking into account quantum fluctuations at all energies, while in the limit $k\to\infty$ one obtains a theory governed by a simple classical action.
If this is the case, we can obtain information about the full theory by analysing how the generating functional and the effective action transform under rescaling of $k$. 

The scale $k$ is usually introduced adding a quadratic contribution in the construction of $Z$. As discussed in the introduction, we use a local regulator
   \[
	Q_{k} = -\frac{1}{2} \int \dd x\, \q_{k}(x) \chi(x)^2\,,
    \]
    and we study the behaviour of $\Gamma$ and $W$ under changes of the scale $k$.
    \begin{remark}\label{Rmk: action I0k}
        We stress that the equation of motion for the action $I_{0k}=I_0+Q_k$ reads $P_{0k}\chi=P_0\chi+Q_k^{(1)}=(P_0-q_k)\chi=(\square-m^2-q_k)\chi=0$, so that $q_k$ really plays the role of a mass term.
        Moreover, to avoid any confusion, we stress that, in what follows, $\mathcal{A}$ will denote the $\ast$-algebra associated with the action $I_0$.
    \end{remark}
	In this paper we chose $\q_k(x) = k^2 f(x)$, where $f$ is a compactly supported smooth function which is usually $1$ on large region of the spacetime and which plays the role of adiabatic cutoff.
	Eventually, this cutoff is removed considering a suitable limit in which $f$ tends to 1.
	
	Under that limit $Q_{k}$ coincides with a mass contribution to the field, and since usually massive fields show a better infrared behavior compared to massless ones, $Q_k$ plays the role of an infrared regulator.

    We then propose the following definition for the regularised generating functional $Z_{k}$:
    \begin{equation}\label{eq:Z-reg}
    Z_{k}(j) := \omega(S(V)^{-1} \star S(V + J + Q_{k})) \,,
    \end{equation}
    which reduces to the $Z(j)$ given in Equation \eqref{Eq: Z-generating functional} in the limit $k\to0$, since $Q_k$ vanishes.
    Here, $\omega$ is an arbitrary Hadamard state on $\mathcal{A}$ which is not necessarily quasifree for the free theory.

This regularization is consistent with the usual IR regularization one can find in the literature \cite{Berges2000}. Actually, if the Gell-Mann-Low formula holds, we can factor the definition of the relative partition function into
\begin{equation}
Z_k(j) = \frac{1}{\omega(S(V))} \omega(S(V+J+Q_k))
{=:}\frac{\mathcal{Z}_k(j)}{\omega(S(V))}\ ,
\end{equation}
which, apart from a normalization constant $\omega(S(V))$, coincides with the usual regularised path integral formulation.
However, since the Gell-Mann-Low formula is broken on a generic curved spacetime $\mathcal{M}$ or if the state $\omega$ is not the Minkowski vacuum, we shall derive the generating functional for the connected correlation functions and the effective action starting from $Z_k(j)$ instead of $\mathcal{Z}_k(j)$.

In analogy with the unregularised functionals, we define the regularised generating functional for the connected correlation functions as
\begin{equation}\label{Eq: Wk definition}
    W_{k}(j) = - i \log Z_{k}(j) \ .
\end{equation}
The \textbf{classical field} $\phi$ at fixed current $j\in C_{\mathrm{c}}^\infty(M)$ is then obtained as
\begin{equation} \label{def-phi}
    \frac{\delta W_k}{\delta j(x)} = \frac{1}{Z_k(j)} \omega \big ( R_V(S(J+Q_k)\cdot_T \chi(x) )\big ) = \phi(x) \ .
\end{equation}
As we show in Proposition \ref{Prop: existence of jphi}, the relation between $j$ and $\phi$ can be inverted to get 
the current
$j=j_\phi$ which solves \eqref{def-phi} as a function of $\phi$, at least in the sense of perturbation theory.
Hence the Legendre transform can be applied to $W$  and it gives 
\begin{equation} \label{def-tilde-gamma}
\tilde{\Gamma}_{k}(\phi) = W_{k}(j_\phi) - J_\phi(\phi) \,.
\end{equation}
where $J_\phi(\chi) =\Chi_{j_\phi}(\chi)= \int \dd x j_\phi(x)\chi(x)$. 
Finally, we can translate $\tilde \Gamma_{k}$ to get the \textbf{average effective action},
\begin{equation}\label{Eq: average effective action}
\Gamma_{k}(\phi) = \tilde \Gamma_{k}(\phi) - Q_{k}(\phi) \,.
\end{equation}
By definition of the Legendre transform, the derivative of $\tilde \Gamma_k$ gives the \textbf{quantum equations of motion} \begin{align} \label{qeom}
\frac{\delta \tilde \Gamma_k}{\delta \phi} = \frac{\delta (\Gamma_k+Q_k)}{\delta \phi}= - j_\phi \ .
\end{align}
Hence, from \eqref{qeom} and \eqref{def-phi}, we have
\begin{align}
\nonumber
\delta(x,y) &= \frac{\delta j_\phi(x)}{\delta j_\phi(y)} = - \frac{\delta}{\delta j_\phi(y)} \frac{\delta }{\delta \phi(x)} (\Gamma_{k} + Q_{k}) \\
\label{relation-second-derivatives}
&= 
- \int\dd z\frac{\delta \phi(z)}{\delta j_\phi(y)} \frac{\delta}{\delta \phi(z)}\frac{\delta}{\delta \phi(x)} (\Gamma_k +Q_k) = -\int\dd z (\Gamma^{(2)}_{k} +Q_{k}^{(2)})(x,z) \frac{\delta^2 W_{k}}{\delta j(z)\delta j(y)} \ ,
\end{align}
showing that $\Gamma_k^{(2)}+Q_k^{(2)}=\Gamma_k^{(2)} - q_k$ is (minus) the inverse of the interacting propagator.

By direct computation, we find that the second functional derivative of $W_k$ is
\begin{multline} \label{interacting-F-propagator}
- i \frac{\delta^2 W_{k}(j)}{\delta j(x)j(y)} = \frac{1}{Z_{k}(j)}\omega \Big (S(V)^{-1} \star [S(V+J+Q_{k}) \cdot_T \chi(x)  \cdot_T \chi(y)] \Big )  \\
- \frac{1}{Z_{k}(j)^2} \omega \Big (S(V)^{-1} \star [S(V+J+Q_{k}) \cdot_T \chi(x)] \Big )
\ \omega \Big (S(V)^{-1} \star[ S(V+J+Q_{k}) \cdot_T \chi(y)] \Big ) \\
= \frac{1}{Z_{k}(j)}\omega \Big (S(V)^{-1} \star [S(V+J+Q_{k}) \cdot_T \chi(x)  \cdot_T \chi(y)] \Big ) - W_k^{(1)}(x) W_k^{(1)}(y) \ .
\end{multline}
In section \ref{sec:free-case}, using the principle of perturbative agreement, \textit{cf.} Appendix \ref{sec:perturbative-agreement}, we will see that the second derivative of $W_k$ is, in certain limits, 
the Feynman propagator for the regularised theory.

We finally notice that, as in the $k$-independent case, for finite $j$ as well as finite $k$, the functional derivatives of $W_{k}(j)$ are not the connected correlation functions; only taking the limits $j \to 0$ and $k \to 0$ one recovers the meaning of $W_{k}$ as a generating functional of the truncated time-ordered correlation functions.

\begin{remark}
One of the basic requirements for a well-defined $S$-matrix is that it is unitary.
It follows that $Z_k(j)$ is pure phase because it is the expectation value of a product of unitary operators, thus $W_k(j):=-i\log Z_k(j)$ must be real.
This in turn ensures the reality of the average effective action, implying that quantum contributions to the action cannot give rise to complex couplings.
\end{remark}

\subsection{Properties of the average effective action}

    \subsubsection{Quantum equations of motion}
    
Before applying our formalism to a particular example, we want to study some general properties of the average effective action $\Gamma_k$. First of all, we want to show that $\tilde \Gamma_k$ is well-defined, in the sense that the relation $\phi = \phi(j)$ given by equation \eqref{def-phi} is invertible into $j_\phi = j_\phi(\phi)$ in the sense of perturbation theory ---see \cite{Ziebell2021} for a non perturbative version of this result on Euclidean space.
We have the following proposition:
\begin{proposition}\label{Prop: existence of jphi}
		Let $\omega$ be a state on $\mathcal{A}$ and let $\phi\in C^\infty(M)$ be such that $\phi=i \Delta_F \tilde{j}_0 $, for some $\tilde{j}_0\in C_{\mathrm{c}}^\infty(M)$, while $\Delta_F$ is the Feynman propagator of the theory we are considering.
		Then, outside of the adiabatic limit, there exists a unique $j_\phi\in C_{\mathrm{c}}^\infty(M)[[V]]$ which solves Equation \eqref{def-phi}.
		Furthermore, $j_\phi$ can be obtained solving 
\begin{equation} \label{eq-j}
j = - P_0 \phi - Q^{(1)}_k(\phi) - \frac{1}{Z_k(j)} \omega\big(S(V)^{-1} \star [S(V+ Q_k + J) \cdot_T TV^{(1)} ]\big) \ 
\end{equation}
	by induction on the perturbation order.
	\end{proposition}
	\begin{proof}
		We recall that $P_0\phi = I_0^{(1)}(\phi)$.
		By applying $P_0$ on both sides of Equation \eqref{def-phi}, we obtain
\begin{equation} \label{intermediate-1}
P_0 \phi = \frac{1}{Z_k (j)} \omega \big (S(V)^{-1}\star[ S(V+J+Q_k) \cdot_T P_0 \chi ]\big ) \ .
\end{equation}
In fact, since $P_0$ is a partial differential operator, it acts only on the spacetime-dependent quantities, i.e. $\chi$. Now, we can turn the $T$-product into a $\star$-product, because, if $B$ is linear in the field configurations, it holds that 
\begin{equation}
A \cdot_T B = A \star B + i \int A^{(1)}(x) \Delta_A(x,y) B^{(1)}(y) \dd x \dd y \ ,
\end{equation}
given the relation $\Delta_F = \Delta_+ + i \Delta_A$ and the definitions of the products. Applying the above relation in \eqref{intermediate-1}, and recalling that $P_0 \Delta_A = \delta$, we obtain
\[
P_0 \phi(x) = \frac{1}{Z_k (j)} \bigg( \omega \big (S(V)^{-1}\star S(V+J+Q_k) \star P_0 \chi(x) \big) + i \omega \big (S(V)^{-1}\star S(V+J+Q_k)^{(1)}(x)\big)\bigg) \ .
\]
The first term in parenthesis vanishes, because $\omega$ satisfies the free equation of motion and thus  $\omega( A \star P_0 \chi) = 0$, while the derivative of $S$ in the second term can be computed explicitly, leading to
\[
P_0 \phi  = - \frac{1}{Z_k(j)} \omega\big ( S(V)^{-1} \star [ S(V+ Q_k + J) \cdot_T (TV^{(1)} + Q^{(1)}_k + J^{(1)} ) ]\big ) \ .
\]
Since $J$ is linear in the field $\chi$, its first derivative gives the classical current $j(x)$. 
On the other hand, we have $Q_k^{(1)}(\chi)(x) = -\q_k(x) \chi(x)$, and thus 
\begin{align*}
\frac{1}{Z_k(j)}\omega\big(S(V)^{-1}\star  [S(V+ Q_k + J) \cdot_T Q^{(1)}_k(\chi)] \big )
&=-q_k\frac{1}{Z_k(j)}\omega\big(S(V)^{-1}\star  [S(V+ Q_k + J) \cdot_T \chi ]\big ) 
\\&=-q_k\phi
= Q_k^{(1)}(\phi) \ .
\end{align*}
This shows that Equation \eqref{def-phi} is equivalent to Equation \eqref{eq-j},
which can be used to obtain $j_\phi$ from $\phi$ as a formal power series in $V$.

Notice that the obtained solution is unique and lies in $C_{\mathrm{c}}^\infty(M)[[V]]$.
In fact, at zeroth order in perturbation series, the equation simply gives
\begin{equation}
j_{\phi,0}(x) 
{= - P_0\phi(x)  - Q^{(1)}_k(\phi)(x)
= - (P_0 - q_k) \phi(x) \ ,}
\end{equation}
which is nothing but the free, regularised equations of motion. 
Furthermore, $j_{\phi,0}\in C^\infty_\mathrm{c}(M)$ because $P_0\phi = -\tilde{j}_0\in C_{\mathrm{c}}^\infty(M)$ by hypothesis, and outside the adiabatic limit $\q_k$ is smooth and of compact support.
Proceeding by induction, we see that if $j_\phi$ is compactly supported up to order $V^{n-1}$, then so is up to order $V^n$  because, denoting by $j_{\phi,{n}}$ the solution up to order $V^{n}$,
\[
j_{\phi,n} 
=
j_{\phi,0} - \frac{1}{Z_k(j_{\phi,n-1})} \omega\big(S(V)^{-1}\star(  S(V+ Q_k + J_{\phi,n-1}) \cdot_T TV^{(1)}) \big) \ 
\]
and outside the adiabatic limit $TV^{(1)}$ is also of compact support.
\end{proof}

\begin{remark}\label{remark-semiclassical-limit}
We observe that in the limit $k\to0$ the previous proposition implies the well-posedness of the Legendre transform of $W(j)$ to $\tilde{\Gamma}(\phi)$ also in the unregularised case.  
Furthermore, equation \eqref{eq-j} of
Proposition \ref{Prop: existence of jphi} is nothing but the quantum equation of motion \eqref{qeom}. 
It can be used to obtain the form of the effective action $\Gamma_k$. In particular, at linear order in $V$ and in the limit $k\to 0$, using Lemma \ref{Lem: time-ordered multiplication with S(J)},  Equation \eqref{eq-j} reduces to
		\begin{align*}
			j_\phi&=
			-P_0\phi
			-\frac{1}{\omega(S(J_\phi))}\omega\left(
			S(J_\phi)\cdot_T TV^{(1)}
			\right)
			\mod O(V^2)
			\\&=-P_0\phi
			-TV^{(1)}(i\Delta_Fj_\phi)
			\mod O(V^2)
			\\&=-P_0\phi
			-TV^{(1)}(\phi)
			\mod O(V^2)\,.
		\end{align*}
Recalling that $\dfrac{\delta \tilde\Gamma}{\delta \phi} = -j_\phi$ 
we have that, up to normal ordering, at leading order the effective action coincides with the classical action $I$.
\end{remark}

\subsubsection{Classical limit} \label{sec:classical-limit}
From Equations \eqref{eq-j} and \eqref{Eq: average effective action}, substituting the quantum equations of motion \eqref{qeom}, we get
\begin{equation} \label{qeom-explicit}
\Gamma_k^{(1)}(\phi) = P_0 \phi + \frac{1}{Z_k(j_\phi)} \omega\big(R_V (S(Q_k + J_\phi) \cdot_T TV^{(1)}) \big) \ .
\end{equation}
We would like to compute the limit $k \to \infty$ from the above equation. In the Euclidean case, this limit has been discussed in some detail in \cite{Reuter1996-Liouville}, where it was shown that the average effective action and the bare action, in the $k\to \infty$ limit, differ by the infinite mass limit of a one-loop determinant.

To study the limit in our formalism, it is convenient to apply the perturbative agreement \cite{HollandsWald2004,Drago2015,Zahn2015}, which we recall in Appendix \ref{Sec: Principle of Perturbative Agreement}, \textit{cf.} Equation \eqref{Eq: gPPA for Moeller operators} and Equation \eqref{Eq: gPPA for relative S-matrices} in particular.
This way we may convert the non-commutative products $\star$ to $k$-dependent products $\star_k$.
In fact, at least when the state $\omega$ is assumed to be quasifree for the free theory, we have the following

\begin{lemma}\label{lem:step}
Let $\omega$ be a quasifree state on $\mathcal{A}$ (associated to $I_0$), and consider $\star$ and $\cdot_T$ constructed out of the two-point function $\Delta_+$ of $\omega$.
Let $\star_k$ and $\cdot_{T_k}$ the star and time ordered products of $\mathcal{A}_k$, which descends from the action $I_0+Q_k$, and constructed out of $\Delta_{+,k} = \mathsf{r}_{Q_k} \Delta_+ \mathsf{r}_{Q_k}^*$ with $\mathsf{r}_{Q_k}$ given in Equation \eqref{Eq: classical Moeller maps} of Appendix \ref{Sec: Principle of Perturbative Agreement}  and $\mathsf{r}_{Q_k}^*f:=f-q_k\Delta_{A,k}f$. 
It holds that for every $A\in\mathcal{F}_{\textrm{\normalfont loc}}$
\begin{align}\nonumber
\omega\big(R_V (S(Q_k + J)\cdot_T A) \big) 
&= \eval{S(V)^{-1} \star [S(V+Q_k+J)\cdot_T A]}_{\chi=0} \\
\label{Eq: variation of expectation value using gPPA}
&= \eval{S_k(\gamma_k V- \gamma_k Q_k)^{-1}\star_k [S_k(\gamma_k V+J)\cdot_{T_k} \gamma_kA]}_{\chi=0} \ .
\end{align}
where $S_k$ is the $S$-matrix constructed with $\cdot_{T_k}$ and where
 the map $\gamma_k$ is given in \eqref{Eq: PPA map}.
\end{lemma}

\begin{proof}
Up to a constant factor, the left hand side of 
\eqref{Eq: variation of expectation value using gPPA} can be obtained evaluating 
$S(V)^{-1} \star S(V+Q_k+J+{\mu} A)$ on the state $\omega$ and taking the derivative with respect to $\mu$ in $\mu=0$. Hence 
we start rewriting $S(V)^{-1} \star S(V+Q_k+J+{\mu} A)$, inserting the identity $1 = S(V+Q_k) \star S(V+Q_k)^{-1}$:
\begin{align*}
S(V)^{-1} \star S(V+Q_k+J+{\mu} A) &= S(V)^{-1} \star S(V + Q_k) \star S(V+Q_k)^{-1} \star S(V+Q_k+J+{\mu} A) 
\\ &= R_V(Q_k) \star S_{V+Q_k}(J+{\mu} A) \ .
\end{align*}
The first factor can be rewritten using Equation \eqref{Eq: gPPA for Moeller operators} of Appendix \ref{Sec: Principle of Perturbative Agreement} as
\[
R_V(Q_k) = \RC_{Q_k}  R_{k,\gamma_k (V-Q_k)}(\gamma_k Q_k).
\]
The second term can be rewritten using Equation \eqref{Eq: gPPA for relative S-matrices} of Appendix \ref{Sec: Principle of Perturbative Agreement} as
\[
S_{V+Q_k}(J+{\mu} A) = \RC_{Q_k}\left[ S_k(\gamma_k V)^{-1} \star_k S_k(\gamma_k V+J+\gamma_k {\mu} A)\right] \,,
\]
where we also used the identity $\gamma_kJ=J$ because $J$ is $\chi$-linear.
Since $\RC_{Q_k}$ intertwines $\star$ to $\star_k$ and since $\eval{\RC_{Q_k}B}_{\chi=0} =
\eval{B\circ \mathsf{r}_{Q_k}}_{\chi=0}
=
\eval{B}_0$
because $\eval{\mathsf{r}_{Q_k}(\chi)}_{0}=0$, equation \eqref{Eq: variation of expectation value using gPPA} follows.
\end{proof}

\begin{theorem}\label{Thm: classical limit of Gamma on ultrastatic spacetimes}
Let $(\mathcal{M},g)$ be an ultrastatic spacetime with bounded curvature, let $\omega$ be the ground state on $\mathcal{A}$, and
consider the limit where the support of $q_k$ tends to $\mathcal{M}$, namely where $q_k=k^2$.
Then the average effective action $\Gamma_k$ coincides with the classical action up to a constant in the limit where $q_k=k^2/2$ and  $k\to\infty$, namely, 
it holds that
\[
\Gamma_k^{(1)}(\phi)\underset{k\to \infty}{\longrightarrow} I^{(1)}(\phi)
\]
in the sense of pointwise converges of functions at any order 
in the coupling constant.
\end{theorem}
\begin{proof} We start with the quantum equation of motion given in the form \eqref{qeom-explicit}
\[
\Gamma_k^{(1)}(\phi) = P_0 \phi + \frac{1}{Z_k(j)} \omega\big(R_V (S(Q_k + J) \cdot_T T V^{(1)}) \big) \ .
\]
Using Lemma \ref{lem:step}, we may rewrite it as 
\[
\Gamma_k^{(1)}(\phi) = P_0 \phi + \frac{1}{Z_k(j)} \omega_k\bigg(S_k(\gamma_k V - \gamma_k Q_k )^{-1} \star_k [S_k(\gamma_k V +J)\cdot_{T_k} \gamma_k T V ^{(1)}] \bigg) \ 
\]
where $\omega_k=\omega\circ r_{Q_k}$.

\color{black}
When $q_k$ tends to $k^2$, $\omega_k$ tends to the ground state related to the equation $P_0\phi-k^2\phi=0$ (The proof for the  case of Minkowski background can be found in Lemma D.1 in \cite{Drago2015} and it can be generalized to generic ultrastatic spacetime with bounded scalar curvature.
See also \cite{Drago2018} for the case of equilibrium states.)

The spacetime $\mathcal{M}$ is ultrastatic, hence, it admits a natural notion of time, by means of which, 
$\mathcal{M}=\mathbb{R}\times \Sigma$. Furthermore, 
\[
P_0 = \partial_t^2 - B-k^2
\]
where $B$ is a self-adjoint operator on $L^2(\Sigma)$ whose spectrum is bounded from below (by $m^2-\xi \| R \|_\infty$). Hence if $k$ is sufficiently large, $-k$ is in the resolvent set of $B$ and thus for large $k$  $(B+k^2)^{-1}$ is a bounded positive operator.  
We furthermore observe that for any $\mathbb{N}\ni l>0$ and $k^2 > r$, we can define by spectral calculus $(B+k^2)^{-l}$ and
\[
\|(B+k^2)^{-l} \psi\|_2 \leq \frac{1}{(k^2-r)^{l}}\|\psi\|_2, \qquad \psi \in L^{2}(\Sigma),  
\]
where $r$ is a positive constant which is such that $(m^2-\xi R)\geq -r$ uniformly on $\mathcal{M}$. 

With this at disposal, we can now construct 
the operators 
$\tilde\Delta_{+,k}(t)$ and $\tilde\Delta_{k,F}(t)$
used as integral kernels of $\Delta_{F,k}$ and of $\Delta_{+,k}$ 
by standard functional calculus
over the spectrum of $B+k^2$, which is contained in $\mathbb{R}^+$ for sufficiently large $k$.
Hence, for every $t$
\[
\tilde\Delta_{+,k}(t)=
\frac{e^{i t \sqrt{B+k^2}}}{2\sqrt{B+k^2}}
\]
and 
\[
\tilde\Delta_{F,k}(t):= \theta(t)\Delta_{+,k}(t)
+\theta(-t)\tilde\Delta_{+,k}(-t)
\]
and both are elements of $B(L^{2}(\Sigma, \dd x)$ and their operator norms are such that
\begin{equation}\label{eq:boundDelta+k}
\|\tilde\Delta_{+,k}(t)\| \leq \frac{1}{2\sqrt{k^2-r}}
,
\qquad
\|\tilde\Delta_{F,k}(t)\| \leq \frac{1}{2\sqrt{k^2-r}}.
\end{equation}
With these two operators at disposal we have that for every $h,g\in C^\infty_0(\mathcal{M})$.
\[
(\Delta_{+,k}g)(t_x,\mathbf{x}) = 
    \int_{\mathbb{R}} \dd t' \bigg(\tilde{\Delta}_{+,g}(t_x-t')g(t',\cdot)\bigg)(\mathbf{x}) 
\]
and similarly for $\Delta_{F,k}$.
With this observation and the estimates of $\tilde{\Delta}_{+,k}$ and $\tilde{\Delta}_{F,k}$ at disposal valid uniformly in time, we can estimate the distributions $\Delta_{+,k}^{\otimes n}$ and $\Delta_{F,k}^{\otimes n}$ on $\mathcal{M}^{2n}$.

Operating as in Lemma \ref{le:limit-propagators}  we can observe that for every $h,g\in C^{\infty}_c(\mathcal{O})$ and with $f\in C^{\infty}_c(\mathcal{M})$ which is $1$ on $\mathcal{O}$ and for every $l\in\mathbb{N}$ with $C$ a suitable constant
\begin{equation}\label{eq:chain-est}
|\langle h,\Delta_{+,k} g \rangle| \leq 
|\langle h,\frac{P_0^l}{k^{2l}}\Delta_{+,k} g\rangle |
\leq \frac{1}{k^{2l}}\|P_0^l h\|_2 \| f \Delta_{+,k} g \|_2 \leq \frac{C}{k^{2l}}\|P_0^l g\|_2  \|f\|_2 \|g\|_2
\end{equation}
where now the $\|\cdot\|_2$ norms act on $L^2(\mathcal{M},\dd x)$, and where we used the fact that $\Delta_{+,k}$ is a weak solution of $P_0-k^2$ in the first inequality, Cauchy-Schwartz inequality in the second step and the uniform estimates \eqref{eq:boundDelta+k} in the third one. 

We can now generalize this observation to estimate
\begin{equation}\label{eq:est-Deltak-tensor}
|(F^{(n)} \Delta_{+,k}^{\otimes n} G^{(n)})| \leq \frac{C_l}{k^{2l}}
\end{equation}
for $F$,$G$ in which are obtained as tensor product of local functionals and
valid for large $k$ and for every $l\in\mathbb{N}$.
Actually, we observe that as an operator on 
$L^2(\Sigma)\otimes L^2(\Sigma)$
\[
\frac{1}{i(\sqrt{B_1+k^2}+\sqrt{B_2+k^2})}\partial_{t_k} \Delta_{+,k}(t_{x}-t_{y})\Delta_{+,k}(t_{x}-t_{z}) = \Delta_{+,k}(t_{x}-t_{y})\Delta_{+,k}(t_{x}-t_{z})
.
\]
Using this and other similar observations in estimates analogous to \eqref{eq:chain-est} we get the desired \eqref{eq:est-Deltak-tensor}.
With this at disposal, we observe that the $\star_k$-product reduces to the point-wise product in the limit $k\to\infty$ even if in one of the factors $Q_k$ appears.

Using also Lemma \ref{Lem: time-ordered multiplication with S(J)} we actually 
get that
\begin{multline*}
\lim_{k\to\infty} 
\frac{\omega_k\bigg(S_k(\gamma_k V - \gamma_k Q_k )^{-1} \star_k [S_k(\gamma_k V +J)\cdot_{T_k} \gamma_k T V ^{(1)}] \bigg)}{Z_k(j)} 
=
\\
\lim_{k\to\infty}
\frac{ \omega_k\bigg(S_k(\gamma_k V_{\phi_0})\cdot_{T_k} \gamma_k T V ^{(1)}_{\phi_0} \bigg)
}{\omega_k\bigg(S_k(\gamma_k V_{\phi_0})\bigg)}
\end{multline*}
where $F_{\phi_0}(\chi) = F(\chi+\phi_0)  $  and where $\phi_0 = i \Delta_{F,k} j_\phi $.
We observe now that in the limit $k\to \infty$, thanks to the estimate \eqref{eq:boundDelta+k} the
$T_k$-product among local functionals reduces to a pointwise product.
\color{black}
Furthermore, $\gamma_k T V^{(1)}=T_k V^{(1)}$ ---\textit{cf.} Remark \ref{Rmk: T, TLambda Wick ordering} in Appendix \ref{Sec: Principle of Perturbative Agreement}--- and under the same limit $T_k V^{(1)}$ tends to $V^{(1)}$. 
Finally, using \eqref{eq-j}
we get that in the limit $k\to\infty$, $\phi_0$ converges to $\phi$, hence, since $V$ as a function of $\phi$ is smooth
we obtain
\[
\lim_{k\to\infty}
\frac{ \omega_k\bigg(S_k(\gamma_k V_{\phi_0})\cdot_{T_k} \gamma_k T V ^{(1)}_{\phi_0} \bigg)
}{\omega_k\bigg(S_k(\gamma_k V_{\phi_0})\bigg)} = V^{(1)}( \phi)
\]
and hence
\[
\Gamma_k^{(1)}(\phi) 
\xrightarrow[k\to\infty]{}  P_0 \phi + V^{(1)}( \phi) = I^{(1)}(\phi) \ ,
\]
where $I(\phi)$ is the classical action and the limit holds in the sense of pointwise convergence of functions.
\end{proof}

Although Theorem \ref{Thm: classical limit of Gamma on ultrastatic spacetimes} is proved for the case of ultrastatic spacetime with bounded curvature and for ground state, its thesis holds in a more general setup.
Indeed, the generalization to the case of states which satisfy a similar bound as those given in \eqref{eq:boundDelta+k} is straightforward. Notice that equilibrium states on flat spacetimes or Bunch Davies states in the case of de Sitter backgrounds satisfy a similar estimate.
This shows that the average effective action and the classical action coincide, up to a constant, in the limit $k \to \infty$; more precisely, the expansion at $k \to \infty$ coincides with the semiclassical approximation described in Remark \ref{remark-semiclassical-limit}.

\subsection{Parity of $V$}
Equation \eqref{qeom-explicit} is the starting point for a perturbative study of the average effective action. In this section we want to analyse the quantum corrections to the classical approximation of the effective action, and in particular we want to see if the quantum corrections can violate the symmetries of the classical action.

Since, in this paper, we will only study the $O(1)$ scalar model, we study the parity of the effective average action. In particular, if $V$ is even, $\chi \to - \chi$ is the only symmetry of the classical action.

\begin{proposition}\label{pr:parity}
Let $V$ be even with respect to $\chi \to - \chi$, so that it contains an even number of fields only. 
Then, if $\omega$ is quasifree, the average effective action $\Gamma_k$ is also even.
\end{proposition}

\begin{proof}
    Since $Q_k$ is $\phi$-even, to study the parity of $\Gamma_k$ it suffices to study the parity of $\widetilde \Gamma_k$. We analyse the $\phi$-parity of $\widetilde{\Gamma}_k$ through Equation \eqref{def-tilde-gamma}:
    \begin{align*}
        \widetilde{\Gamma}_k(\phi)
        =W_k(j_\phi)-J_\phi(\phi)
        =-i\log Z_k(j_\phi)-J_\phi(\phi)\,,
        \qquad
        Z_k(j_\phi)=\omega\left[
        R_V(S(Q_k+J_\phi)
        \right]\,.
    \end{align*}
    To this end we first analyse the $\chi$-parity of the involved functionals.
    Since $V$ is $\chi$-even and $\star,\cdot_T$ preserve $\chi$-parity, it follows that $R_VF$ has the same $\chi$-parity of $F$.
    Since $Q_k$ is $\chi$-even and $\omega$ corresponds to evaluation at $\chi=0$, it follows that the contribution
    \begin{align*}
        \omega\left[
        R_V(S(Q_k+J_\phi))\right]
        =\omega\left[
        R_V(S(Q_k)\cdot_TS(J_\phi))\right]\,,
    \end{align*}
    contains only even powers of $J_\phi$, so this contribution is $j_\phi$-even.
    It then follows that $W_k(j_\phi)=-i\log Z_k(j_\phi)$ is $j_\phi$-even.
    
    We now prove that $j_\phi$ is $\phi$-odd. With the observations already made, this will imply that $J_\phi(\phi)$ is $\phi$-even and so is $\widetilde{\Gamma}_k$, because of Equation \eqref{def-tilde-gamma}.
    From Equation \eqref{eq-j}
    we have
    \begin{align*}
        j_\phi = - P_0 \phi - Q^{(1)}_k(\phi) - \frac{1}{Z_k(j_\phi)} \omega\left[R_V\left(S(Q_k + J_\phi) \cdot_T TV^{(1)} \right)\right] \,.
    \end{align*}
    It follows that $j_{\phi,0}$ (the 0-th order in $V$ of $j_\phi$) is $\phi$-odd.
    By induction, let assume that the expansion $j_{\phi,n}$ up to order $V^n$ is $\phi$-odd.
    For the expansion $j_{\phi,n+1}$ up to $V^{n+1}$ we have
    \begin{align*}
        j_{\phi,n+1}
        =-\frac{1}{Z_k(j_{\phi,n})} \omega\left[{R_V}\big(S(Q_k+J_{\phi,n})\cdot_T TV^{(1)}\big) \right] \,,
    \end{align*}    
    which is shown to be $\phi$-odd.
    Indeed, $Z_k(j_{\phi,n})$ is $j_{\phi,n}$-even and thus $\phi$-even, moreover, \begin{align*}
        \omega\left[R_V\big(S(Q_k+J_{\phi,n})\cdot_T TV^{(1)}\big) \right]\,,    
    \end{align*}
    contains only odd powers of $j_{\phi,n}$, because only $\chi$-even terms in $S(Q_k+J_{\phi,n+1})\cdot_T TV^{(1)}$ provide a non-vanishing contribution to the expectation value ---notice that $TV^{(1)}$ is $\chi$-odd because so is $V^{(1)}$ and also because the map $T$ preserves the $\chi$-parity.
    By the inductive assumption, we find that $j_{\phi,n+1}$ is $\phi$-odd.
\end{proof}

Hence, we see that quantum contributions cannot violate the parity of symmetry of the starting classical action, at least in this simple example. This observation will justify the ansatz \eqref{explicit-effective-potential} in actual computations of the flow.

\section{Flow equations} \label{section:flow-eq}
In this section, we derive the flow equations for the generating functionals as differential equations in the scale parameter $k$. We are interested in particular in deriving the equations for $W_k$ and $\Gamma_k$, and in  seeing the connection with their Euclidean counterparts given respectively by the Polchinski \cite{Polchinski1983} and Wetterich \cite{Wetterich1992} equations.


The flow equation for $W_{k}$ can be computed from its definition, since the $k$-dependence comes through the term $Q_k(\chi)= -\frac{1}{2}\int \dd x \q_{k}(x) \chi^2(x)$ in the definition of $Z_{k}$:
\begin{equation} \label{flow-eq-w}
\partial_k W_{k}(j) = -\frac{1}{2} \int \dd x \partial_k \q_{k}(x) \frac{1}{Z_{k}(j)} \omega(S(V)^{-1} \star[S(V+J+Q_{k}) \cdot_T T\chi^2(x)]) \ .
\end{equation}
We recall here that $\omega$ is an arbitrary (not necessarily quasifree) Hadamard state on $\mathcal{A}$.
We shall see in a moment that this equation is a normal-ordered generalization of the Polchinski equation to Lorentzian manifolds and generic states.

From the flow equation for $W_k$, we can immediately write the flow equation for $\tilde \Gamma_{k}$:
\[
\partial_k \tilde \Gamma_{k} = \partial_k W_{k}(j_\phi) + 
{\langle W_{k}^{(1)}(j_\phi), \partial_{k} j_\phi \rangle - \langle \phi,  \partial_{k} j_\phi \rangle} = \partial_{k} W_{k}(j_\phi) \ ,
\]
thanks to \eqref{def-phi}. Subtracting $\partial_k Q_k(\phi)$ to the left-hand side, we get the flow equation for the average effective action $\Gamma_k$ defined in \eqref{Eq: average effective action}:
\begin{equation} \label{flow-eq-gamma}
\partial_k \Gamma_{k} (\phi) = -\frac{1}{2} \int \dd x \partial_k \q_{k}(x) \bigg [ \frac{1}{Z_{k}(j_\phi)} \omega \big (R_V(S(J_\phi+Q_{k}) \cdot_T T\chi^2(x) ) \big ) - \phi^2(x) \bigg ] \ .
\end{equation}
This flow equation, valid for local $Q_k$, is a version of the Wetterich equation that is valid on globally hyperbolic Lorentzian manifolds and for generic states. Actually, analyzing the Bogoliubov map, we can see that the right-hand side of the above equation is closely related to the interacting Feynman propagator; the difference is the factor $S(J_\phi+Q_k)$ in the expectation value. However, applying the perturbative agreement ---\textit{cf.} Equations \eqref{Eq: gPPA for Moeller operators} and \eqref{Eq: gPPA for relative S-matrices} in Appendix \ref{Sec: Principle of Perturbative Agreement}---, we will see in the next section that, at least in the free case, the factor $S(Q_k)$ can be absorbed in the regularised products, so that, taking the limit $j=0$, in which the correlation functions of physical interest are computed, we get the interacting Feynman propagator for the regularised theory.

Equation \eqref{flow-eq-w} can be cast in the form of the Polchinski equation. Remembering that pointwise products are always implicitly normal-ordered with respect to $H$, as was given in \eqref{eq:Hadamard}, we have
\begin{align*}
T\chi^2(x) &= \lim_{y \to x} \chi(x) \star \chi(y) - (H(x,y) + \frac{i}{2}\Delta(x,y))  
\\&= \lim_{y \to x} \chi(x) \cdot_T \chi(y) - ( H(x,y) + 
\frac{i}{2}\Delta_R(x,y) + \frac{i}{2} \Delta_A(x,y) )
\\&=  \lim_{y \to x} \chi(x) \cdot_T \chi(y) - H_F(x,y)  \ ,
\end{align*}
so that the right-hand side of \eqref{flow-eq-w} can be written using \eqref{interacting-F-propagator} as a normal-ordered Polchinski equation:
\begin{align}
\nonumber
\partial_k W_{k} &=-\frac{1}{2} \int \dd x \partial_k \q_{k}(x) \frac{1}{Z_{k}(j)} \omega(S(V)^{-1} \star[S(V+J+Q_{k}) \cdot_T \left[\lim_{y\to x}\chi(x)\cdot_T\chi(y)-H_F(x,y)\right]])
\\
\label{Eq: implicit definition of tilde HF}
&=-\lim_{y\to x}\frac{1}{2} \int \dd x \partial_k \q_{k}(x) \left[\frac{1}{Z_{k}(j)} \omega(S(V)^{-1} \star[S(V+J+Q_{k}) \cdot_T \chi(x)\cdot_T\chi(y)])
-\widetilde{H}_F(x,y)\right]
\end{align}
where $\widetilde{H}_F$ is a counter-term implicitly defined by the last identity.
Equation \eqref{flow-eq-gamma} then becomes
\begin{equation}\label{eq:Wetterich-Local}
\partial_k \Gamma_{k} = \lim_{y \to x}  \frac{i}{2} \int \dd x \partial_k \q_{k}(x) 
\left[\frac{\delta^2 W_{k}(j)}{\delta j(x)\delta j(y)}  - i \widetilde{H}_F(x,y)\right] \ .
\end{equation}
Furthermore, by construction of $\widetilde{H}_F$,  the object in square bracket is at least a continuous function on a sufficiently small neighborhood of the diagonal in $\mathcal{M}\times \mathcal{M}$.

Therefore, using \eqref{relation-second-derivatives}, we can write the flow equation for $\Gamma_{k}$ in terms of its second derivative, arriving at the \textbf{Wetterich equation} for a local regulator:
\[
\partial_k \Gamma_{k} = -\frac{i}{2} \int \dd x \partial_k \q_{k}(x) :\left[ \Gamma^{(2)}_k - \q_{k}\right]^{-1} :_{\widetilde{H}_F} {(x)} \ ,
\]
where $:A:_{\widetilde{H}_F}(x)= \lim\limits_{y\to x}(A(x,y) + i\widetilde{H}_F(x,y))$, similarly to the procedure necessary to evaluate the expectation value of normal-ordered quadratic local fields by point splitting.
Notice that the definition of $\widetilde{H}_F$ is rather implicit but it can be computed when $V$ is quadratic ---\textit{cf.} Section \ref{sec:free-case}.
Notwithstanding we stress that our main focus is in Equation \eqref{flow-eq-gamma}, which should be regarded as the main equation of interest.

 In \cite{Fehre2021}, the use of a local regulator in the flow equation for the graviton propagator implies the introduction of additional counterterms in the Wetterich equation, in order to regularize the UV divergences; in our discussion, such a regularization is provided by the normal-ordering prescription. In this regard, Eq. 6 in \cite{Fehre2021} shares many similarities with the flow equation \eqref{eq:Wetterich-Local} we derived.

Equations \eqref{flow-eq-gamma} and \eqref{flow-eq-w} then are the generalizations, for generic states on Lorentzian manifolds, of the Wetterich and Polchinski equations, respectively. They are interpreted in a different perspective, because a local $Q_k$ is not simply an infrared cut-off but it can rather be interpreted as a perturbative contribution to the mass of the theory.
However, due to the finiteness of the right-hand side of both equations, they give well-defined beta-functions for the parameters.

\subsection{Non-local regulator function}

In the case of a non-local regulator $Q_k = -\frac{1}{2} \int \dd x \dd y  \q_k(x,y) \chi(x)\cdot_T \chi(y)$, 
at least formally, we may get $Z_k$ from \eqref{eq:Z-reg}. 
Proceeding in this way, we obtain flow equations similar to those already known in the literature.
Actually, the derivative of $W_k$ gives
\[
\partial_k W_k = \partial_k W_{k}(j) = -\frac{1}{2} \int \dd x \dd y \partial_k \q_{k}(x,y) \frac{1}{Z_{k}(j)} \omega(S(V)^{-1}\star\left[S(V+J+Q_{k}) \cdot_T \chi(x) \cdot_T \chi(y)\right]) \ .
\]
We can rewrite the above equation as the Polchinski equation,
\[
\partial_k W_{k} =  -\frac{1}{2} \int \dd x \dd y \partial_k \q_{k}(x,y) \bigg[ - i \frac{\delta^2 W_{k}(j)}{\delta j(x)j(y)} + W_{k}^{(1)}(x)W_{k}^{(1)}(y)  \bigg ] \ .
\]
Analogously, the same derivation as before gives the flow equation for the average effective action,
\begin{equation} \label{flow-eq-gamma-nonlocal}
\partial_k \Gamma_{k} (\phi) = \frac{1}{2} \int \dd x \dd y \partial_k q_{k}(x,y) \bigg [ \frac{1}{Z_{k}(j)} \omega \big (R_V(S(J+Q_{k}) \cdot_T \chi(x)\cdot_T \chi(y) ) \big ) - \phi(x)\phi(y))) \bigg ] \ ,
\end{equation}
which can be cast in the form of the Wetterich equation:
\[
\partial_k \Gamma_{k} = -\frac{i}{2} \int \dd x \dd y \partial_k \q_{k}(x,y) \bigg( \Gamma^{(2)}_k - \q_{k} \bigg)^{-1} \ .
\]

The interpretation of the Polchinski and Wetterich equations with a non-local regulator is  the same as in the standard functional renormalization group approach: in particular, equation \eqref{flow-eq-gamma} tells us how the effective average action flows with the scale $k$ of the IR cutoff $Q_{k}$, in terms of the full propagator of the theory. 
An example of the use of this non-local regulator in the case of cosmological spacetimes can be found in the recent work \cite{banerjee2022spatial}.

As they are, equation \eqref{flow-eq-gamma} and its non-local analogue \eqref{flow-eq-gamma-nonlocal} look intractable, and therefore, in order to solve the equation, we need to identify suitable approximation schemes.

\section{Approximations}

In this section we describe a procedure which provides a useful approximation of the effective action $\Gamma_k$ which can be applied when the chosen reference state $\omega$ for the free theory is quasifree.
It simplifies the computation of the effective action and leads to the notion of beta-function associated with the parameters of the theory $I$.

\subsection{Free case} \label{sec:free-case}

In the case $V=0$, the flow equation for local regulator 
reduces to
\[
\partial_k \Gamma_k =  
-\frac{1}{2} \int \dd x \partial_k \q_k(x) \bigg [\frac{1}{Z_k(j)} \omega \big (S(J)\cdot_T S(Q_k) \cdot_T T\chi(x)^2
\big) - \phi(x)^2  \bigg ] \ .
\]
Instead of analyzing this equation directly, we study the form of $\Gamma_k$ in the next proposition.

\begin{proposition}\label{prop:quadratic}
Let $V=0$, let $\omega$ be a quasifree state on $\mathcal{A}$ (associated with $I_0$).
Then the effective action $\Gamma_k$ coincides with $I_{0}$ up to a constant.
Moreover, for $W_k$ defined as in equation \eqref{Eq: Wk definition}, we have
\begin{align}\label{Eq: Wk second derivative in the quadratic case}
    \frac{\delta^2W_k}{\delta j^2}=
    \Delta_{F,k}:=
	\sum_{n\geq 0}(-i)^n
 {(\Delta_F q_k)^{n}}\Delta_F\,,
\end{align}
where in this formula  $\Delta_F$ is the linear operator  associated to the Feynman propagator with the Schwartz kernel theorem and $q_k$ is a multiplicative operator. Furthermore, 
 the series defining $\Delta_{F,k}$ is perturbative in $q_k$.
In particular, the integral kernel of the operator $\Delta_{F,k}$ is the Feynman propagator associated to $I_{0k}$.

\end{proposition}
\begin{remark}
    Notice that the series in Equation \eqref{Eq: Wk second derivative in the quadratic case} should be considered in the sense of perturbation theory with respect to $q_k$.
    \color{black}
    A rigorous construction of 
    $\Delta_{F,k}$ can be obtained using the classical M\o ller map introduced in \cite{Drago2015}.
    In particular, consider
    \[
    \tilde{\Delta}_{F,k} = \Delta_{+,k} + i \Delta_{A,k}
    \]
    where $\Delta_{A,k}$ and $\Delta_{+,k}$ are obtained by means of the functional which realizes the classical M\o ller map 
    $\mathsf{r}_{Q_k}\chi  = \chi - \Delta_{R,k} q_k \chi$, $\chi\in C^\infty(\mathcal{M})$ introduced in Proposition 3.8 in \cite{Drago2015} and recalled in \eqref{Eq: classical Moeller maps} in the appendix. Actually 
    by means of 
    Proposition 3.11 of \cite{Drago2015} 
    $\Delta_{+,k} = \mathsf{r}_{Q_k} \Delta_{+}  \mathsf{r}_{Q_k}^{*} 
    $
    and Lemma 3.10  of \cite{Drago2015} gives that
    $ \Delta_{A,k} = \Delta_{A} \mathsf{r}_{Q_k}^{*}$.
    Furthermore,  using recursively the latter relation 
    we obtain that
     \[
     \Delta_{A,k} = \sum_{j=0}^N \Delta_A (-q_k \Delta_A)^j - 
      \Delta_A (-q_k \Delta_A)^N q_k
        \Delta_{A,k}.
     \]
    With this observation and using the support properties of the advanced and retarded fundamental solutions it is possible to prove that order by order in powers of $q_k$, $\Delta_{F,k} = \tilde{\Delta}_{F,k}$.
    (See the derivation of \eqref{eq:DeltaFk-Moller} below in a similar context for further details on this procedure.)
    \color{black}
    As a matter of fact, the convergence of the series can be addressed in specific scenario, \textit{e.g.} Minkowski or highly symmetric spacetimes.
\end{remark}

\begin{proof}[Proof of Proposition \ref{prop:quadratic}]
We shall prove Equation \eqref{Eq: Wk second derivative in the quadratic case}, which, together with Equation \eqref{relation-second-derivatives}, will lead to the proof of the statement about $\Gamma_k$. 
%
We recall that $W_k^{(1)}(x)=\phi(x)$ and that
\[
- i \frac{\delta^2 W_{k}(j)}{\delta j(x)\delta j(y)} = 
		\frac{1}{Z_{k}(j)}\omega(S(Q_{k}+J)\cdot_T\chi(x)\cdot_T\chi(y))-\phi(x)\phi(y)
\]
	where $Z_{k}=\omega(S(Q_{k}+J))$.
	We study
	\begin{align*}
		\frac{1}{Z_{k}}\omega(S(Q_{k}+J)\cdot_T\chi(x)\cdot_T\chi(y))
		=\Delta_F(x,y)
		+\frac{1}{Z_{k}}\omega(S(Q_{k}+J)\cdot_T\chi(x)\chi(y))\,,
	\end{align*}
    To deal with the contribution $\omega(S(Q_{k}+J)\cdot_T\chi(x)\chi(y))$, we consider the functional
    \begin{align*}
	    \tilde{Q}_f(\chi):=\int\dd x\, [j(x)\chi(x)-q_k(x)\chi(x)^2]
	    +\int\dd x\dd y\, f(x,y)\chi(x)\chi(y)\,,
	\end{align*}
	where $f\in C^\infty_c(\mathcal{M}^2)$.
	With this definition we have
	\begin{align*}
		\frac{1}{Z_{k}}\omega(S(Q_k+J)\cdot_T\chi(x)\chi(y))
		=\frac{\delta}{i\delta f(x,y)}\log
        (S(\tilde{Q}_f))\bigg|_{\substack{f=0\\\chi=0}}\,
	\end{align*}
    because we are working in a representation of $\mathcal{A}$ where $\omega$ correspond to the evaluation on $\chi=0$.
	Since $S(\tilde{Q}_f)=T[\exp (i\tilde{Q}_f)]$ we can evaluate the log-contribution using the linked-cluster theorem\footnote{ See e.g. section 4.3.1 in \cite{Lindner2013} for a derivation of the formula expressing the linked cluster theorem.}:
	\begin{align}\label{eq:linked-cluster}
		S(\tilde{Q}_f)=T(e^{i\tilde{Q}_f})=\exp\left[T^c(e^{i\tilde{Q}_f}_\otimes)\right]\,.
	\end{align}
	Here $T^c$ denotes the connected time-ordered product: shortly, $T^c(A\otimes B)$ corresponds to summing over all contributions of $A\cdot_TB$ which comes from a connected diagram ---\textit{cf.} equation \eqref{Eq: connected time-ordered n-point functions}.
	For example we have, upon renormalization, $T^c(\chi^2(x)\otimes\chi^2(y))=2\chi(x)\Delta_F(x,y)\chi(y)+\Delta_F^2(x,y)$.
	We thus obtain
	\begin{align*}
		\frac{1}{Z_{k}}\omega(S(\tilde{Q}+J)\cdot_T\chi(x)\chi(y))
		&=\frac{\delta}{i\delta f(x,y)}T^c(e^{i\tilde{Q}_f}_\otimes)\bigg|_{\substack{f=0\\\chi=0}}
		\\&=T^c\left(e_\otimes^{iQ_k+iJ}\otimes\chi(x)\chi(y)\right)\big|_{\chi=0}
		\\&=\sum_{n\geq 1}\frac{1}{n!}T^c\left[ (iQ_k+iJ)^{\otimes n}\otimes\chi(x)\chi(y)
		\right]|_{\chi=0}
		\\&=\sum_{n\geq 1}\frac{1}{n!}T^c\left[ (iQ_k)^{\otimes n}\otimes\chi(x)\chi(y)
		\right]|_{\chi=0}
		+\\
		& 
		+\sum_{n\geq 0}\frac{1}{n!}T^c\left[ iJ\otimes iJ \otimes (iQ_k)^{\otimes n}\otimes\chi(x)\chi(y)
		\right]|_{\chi=0}\,.
	\end{align*}
	Using the definition of $T^c$ we have that the graph over which we are summing in evaluating $T^c\left[ (iQ_k)^{\otimes n}\otimes\chi(x)\chi(y)
		\right]$ are all formed by $n+2$ vertices joined by  $n+1$ edges. The external vertices correspond to $\chi(x)$ and $\chi(y)$ the $n$ internal vertices correspond to $Q_k$. All these graphs are equal and differs only by the possible $n!$ permutations of the internal vertices. Hence
	\begin{align*}
\sum_{n\geq 1}\frac{1}{n!}T^c\left[ (iQ_k)^{\otimes n}\otimes\chi(x)\chi(y)
		\right]|_{\chi=0}
		&=\sum_{n\geq 1}((-i)^n
		\Delta_{F}\q_k(\Delta_F\q_k)^{n-1}\Delta_{F}\delta_y)(x) 
		\\&= \Delta_{F,k}(x,y)- \Delta_F(x,y)
		\,,
	\end{align*}
where $\Delta_F$ is considered as a linear operator acting on compactly supported functions and $\delta_y$ is the Dirac delta function centered in $y$.
Notice that $\Delta_{F,k}$ is a Feynman propagator associated to $I_{0k}$.
Indeed, recalling that $P_{0k}=P_0-q_k$ ---\textit{cf.} Remark \ref{Rmk: action I0k} --- we have
	\begin{align*}
	P_{0} \Delta_{F,k} = P_0\sum_{n\geq 0}(-i)^n
		(\Delta_F\q_k)^{n}\Delta_{F} = i  + \q_k \sum_{n\geq 1}(-i)^{n-1}
		(\Delta_F\q_k)^{n-1}\Delta_{F}
		=i  + \q_k \Delta_{F,k}\,.
    \end{align*}
	Furthermore, 
since $J$ is linear in the field, $Q_k$ and $\chi(x)\chi(y)$ are quadratic and since $T^c$ contains only connected components, we have that the
 graph over which we have to sum to evaluate $T^c\left[ iJ\otimes iJ \otimes (iQ_k)^{\otimes n}\otimes\chi(x)\chi(y)
\right]|_{\chi=0}$
are the connected graphs formed by $n+3$ vertices and $n+2$ edges. The external vertices correspond to $J$ and $J$. $n$ internal vertices correspond to $Q_k$ and the remaining vertex to $\chi(x)\chi(y)$. The latter is non local hence,  this graph can be seen as two-graphs, one joining $J$ to $\chi(x)$ and the other joining $J$ to $\chi(y)$. Taking care of the possible internal vertices in this decomposition and on the possible positions of $\chi\chi$ in the original graph we have the following 
	\begin{gather*}
		\sum_{n\geq 0}\frac{1}{n!}T^c\left[ iJ\otimes iJ \otimes (iQ_k)^{\otimes n}\otimes\chi(x)\chi(y)
		\right]|_{\chi=0}\\
		=\sum_{n\geq 0}\frac{1}{n!}T^c\left[ iJ \otimes (iQ_k)^{\otimes n}\otimes\chi(x)
		\right]|_{\chi=0}
		\sum_{l\geq 0}\frac{1}{l!}T^c\left[ iJ \otimes (iQ_k)^{\otimes l}\otimes\chi(y)
		\right]|_{\chi=0} =\phi(x)\phi(y)\,.
	\end{gather*}
	where the last equality can be proven using again the linked cluster theorem \eqref{eq:linked-cluster}, observing that by definition
	\[
			\phi(x) = \frac{1}{Z_{k}}\omega(S(Q_k+J)\cdot_T\chi(x))
		=\frac{\delta}{i\delta j(x)}T^c(e^{i{Q}_k+iJ}_\otimes)|_{\substack{\chi=0}}
		=T^c(e^{i{Q}_k+iJ}_\otimes\otimes \chi(x))|_{\substack{\chi=0}}\,.
	\]
	Summing up we have that
	\begin{align*}
	- i \frac{\delta^2 W_{k}(j)}{\delta j(x)\delta j(y)} 
		=\Delta_{F,k}(x,y)\,.
	\end{align*}
	where the latter Feynman propagator is the $Q_{k}$-massive one.
	The latter equation, together with 
 Equation \eqref{relation-second-derivatives} and Equation \eqref{Eq: Wk second derivative in the quadratic case} entails that $\Gamma_k^{(2)}(\phi)=I_{0k}^{(2)}(\phi)$ for all $\phi$.
	Moreover, since the action $I_{0}$ is even and quadratic in the field and Proposition \ref{pr:parity} ensures that $\Gamma_k$ is even,
	it follows that $\Gamma_k=I_{0}$ up to a constant.
\end{proof}

\subsection{Local potential approximation}
\label{sec:LPA}
In the following, we illustrate how we apply the approximation scheme known in the literature as \textbf{local potential approximation} (LPA) \cite{Dupuis2021} to the Lorentzian case, in order to find approximated solutions of the  the obtained flow equation for the effective action. 

We start analysing the right-hand side of \eqref{eq:Wetterich-Local}.
Firstly we observe that $W^{(2)}_k$ satisfies the following equation
\[
(\Gamma^{(2)}_k - \q_k) W^{(2)}_k = - \delta \ . \ 
\]
However, even if we knew $\Gamma_k$ we would not simply obtain $W^{(2)}_k$ from the equation above for two reasons.
The first difficulty lies in the fact that 
in general $\Gamma_k$ is not quadratic in $\phi$ and thus $(\Gamma^{(2)}_k - \q_k)$ cannot be easily inverted to get $W^{(2)}_k$. 
Furthermore, even when $\Gamma$ is quadratic,
if $\Gamma^{(2)}_k$ contains a principal symbol which is of hyperbolic type, as it would happen in our case, it is known that there exists no unique inverse and thus it is not easy to pick up the inverse corresponding to the expectation value in the state $\omega$ which appears in the definition of $W$.

With this in mind, let us start with the obtained flow equation in the form given in \eqref{flow-eq-gamma}:
\[
\partial_k \Gamma_{k} (\phi) = -\frac{1}{2} \int \dd x \partial_k \q_{k}(x) \bigg [ \frac{1}{Z_{k}(j_\phi)} \omega \big (R_V(S(J_\phi+Q_{k}) \cdot_T T \chi^2(x) ) \big ) - \phi^2(x) \bigg ] \ .
\]
We stress that in the construction we have presented the state $\omega$ has been selected a priori: it does not depend on $k$ and we do not want to change it in the approximation we are going to discuss.
Furthermore, to be as close as possible to the standard literature, the reference state $\omega$ for the free theory is a quasifree Hadamard state.

The approximation scheme proceeds with an Ansatz for the effective average action in the form of a \textbf{local potential approximation},
\begin{equation}\label{eq:Ansatz}
\Gamma_{k}(\phi) =  - \int \dd^d x  \left( \frac{1}{2} \nabla_a\phi \nabla^a \phi  {+U_k(\phi)} \right) \ .
\end{equation}
As a first approximation, we do not consider a $k$-dependent wavefunction renormalization in front of the kinetic term. We will comment on its inclusion in remark \ref{Rmk: wavefunction renormalization} at the end of this section.

With this Ansatz, the left-hand side of the flow equation is 
$\partial_k \Gamma_k = -\int \dd x \; \partial_k U_k(\phi)$.
To approximate the right-hand side in a sensible way, we proceed as follows.
We expand $\Gamma_k$ close to a solution $\phi_{{cl}}$ of $\frac{\delta \Gamma_k}{\delta \phi} = 0$,
\begin{equation*}\label{eq:expansion-gamma-fluctuation-field}
\Gamma_k(\phi) = \Gamma_k(\phi_{cl}) + {\frac{1}{2}\Gamma^{(2)}_k(\phi_{cl})(\varphi,\varphi)} + \mathcal O(\varphi^3)
= \Gamma_k^t(\varphi) + \mathcal O(\varphi^3)
 \ ,    
\end{equation*}
where we introduced the fluctuation field $\varphi := \phi - \phi_{cl}$ and we 
have denoted the \textbf{truncated effective action} by $\Gamma_k^t$.

The local potential approximation proceeds by assuming that $\Gamma_k^t$ is the average effective action produced by an action $I^t_0$, in the very same way $\Gamma_k$ is the average effective action produced by $I=I_0+V$.
By Proposition \ref{prop:quadratic}, the action $I^t_0$ has to be quadratic, because so is $\Gamma_k^t$.
Moreover, from equation \eqref{eq:Ansatz} it follows that the principal symbol of $\Gamma_k^t$ is $k$-independent, so we must have
\[
I^t_0 = - \int \dd x \left( \frac{1}{2} \nabla_a\chi \nabla^a \chi  {+\frac{1}{2}U^{(2)}_k(\phi_{cl})}\chi^2 \right)\, ,
\]
where $U^{(2)}_k(\phi_{cl})$ 
is equal to the contribution to the truncated $\Gamma^t_k$ coming from the potential $U_k$.
We proceed by comparing the "truncated" action $I^t_0$ with the full action $I=I_0+V$ to find an approximation of the right-hand side of \eqref{flow-eq-gamma}.
More precisely, we  decompose the full action $I=I_0+V$ in another free and interacting part, namely 
$I=I^t_0+\mathcal{V}$. We then rewrite the right-hand side of the flow equation perturbing the system around $I^t_0$ and discarding the correction due to $\mathcal{V}$ in the right-hand side of \eqref{flow-eq-gamma}.
We thus have
\[
I=I_0+V = 
I_0^t + \mathcal{V}\,,
\]
where $I_0^t:= I_0+M$ and $\mathcal{V} := V-M$, with 
\[
M :=  -\int \dd x (U^{(2)}_k(\phi_{cl})-  m^2-\xi R )\frac{\chi^2}{2}.
\]
Notice in particular that $M$ depends on the non-trivial field background $\phi_{cl}$.
Furthermore, its role is that of shifting the quadratic non-kinetic parts, originally present in $I_0$, to $-U_k^{(2)}(\phi_{cl})\chi^2/2$.

Truncating $\Gamma$ to $\Gamma^t$ corresponds thus to discarding the $\mathcal{V}$ contributions at the right-hand side of the flow equation \eqref{flow-eq-gamma}, keeping however the choice of the state $\omega$ originally made.
Stated differently, we are approximating the full, non-linear action $I$ with a quadratic one, $I^t_0$, which is the quadratic action whose associated average effective action is $\Gamma_k^t$.
The term $M$ contains the residual information about the interactions of the full Lagrangian.
The action $I_0^t$ then contains the usual kinetic term, plus a non-trivial, classical background, acting as a $\phi_{cl}$-dependent mass for the fluctuation field $\varphi$.

Within this truncation, we now study how the right-hand side of \eqref{flow-eq-gamma} gets modified. In order to proceed, we make use of the principle of perturbative agreement to analyse how the Bogoliubov map changes with this new splitting (see Appendix \ref{sec:perturbative-agreement} for a review).
In fact, we want to use the Bogoliubov map $R_{\mathcal{V}}$ constructed around the new action $I_0^t$ and to consider only the zeroth order contribution.

In particular, making use of Theorem 4.1 in \cite{Drago2015}, we have that ---\textit{cf.} Equation \eqref{Eq: gPPA for Moeller operators}---
\[
R_V = R_{\mathcal{V}+M} = r_M\circ R^M_{\gamma \mathcal{V}}\circ \gamma\,,
\]
where $r_M$ is the classical M\o ller map, whose definition
is recalled in the Appendix \ref{sec:perturbative-agreement} ---\textit{cf.} Equation \eqref{Eq: classical Moeller maps}--- 
and $\gamma$ intertwines the time ordered product $T$ constructed with $\Delta_F$ with a suitably chosen time-ordered product $T_M$ for the free theory $I^t_0$, actually $\gamma T=T_M$.
In particular, the associated Feynman propagator $\Delta_{F,M}$ is the one associated with the quasifree state $\omega_M$, whose two-point function reads $\mathsf{r}_M\Delta_+\mathsf{r}_M^*$ where $\mathsf{r}_M$ is given in Equation \eqref{Eq: classical Moeller maps}. Furthermore, $R^M_{\gamma V}$ is the Bogoliubov map (quantum M\o ller map) constructed over the free theory $I^t_0$. 
Hence, since
\[
R_{\mathcal{V}+M}(S(J_\phi+Q_{k})\cdot_T T\chi^2  ) = r_M \left( R^M_{\gamma \mathcal{V}}(S^M(\gamma Q_k+\gamma J_\phi)\cdot_T \gamma T \chi^2) \right)
\]
holds, discarding the contributions containing $\mathcal{V}$ at the right-hand of \eqref{flow-eq-gamma} gives the flow equation in the relevant approximation. Up to an integration over the space it is given by:
\[
\partial_k {U_k(\phi)} = 
\frac{\partial_k \q_k}{2} \left(\frac{\omega_M (S^M(\gamma Q_k + J_\phi)\cdot_{T_M} T_M \chi^2)} {\omega_M (S^M(\gamma Q_k + J_\phi))} - \phi^2 \right),
\]
where we used the fact that $\gamma T \chi^2=T_M \chi^2$.
The right-hand side of the previous equation can be computed following the analysis presented in Proposition \ref{prop:quadratic}, to obtain
\[
\partial_k U_k {(\phi)}  =  \lim_{y\to x}\frac{\partial_k \q_k}{2} \left(\frac{\omega_M (S^M(\gamma Q_k + J_\phi)\cdot_{T_M}  \chi(x)\cdot_{T_M} \chi(y) )} {\omega_M (S^M(\gamma Q_k + J_\phi))} - H_{F,M,k}(x,y) - \phi(x)\phi(y) \right),
\]
where we used the fact that $\gamma \chi = \chi$.
Hence
\begin{equation} \label{approximate-Wetterich}
\partial_k U_k{(\phi)} =  \lim_{y\to x}  \frac{\partial_k \q_k}{2}(x) (\Delta_{F,M,k}(y,x)-H_{F,M,k}(y,x)) \,,
\end{equation}
where $\Delta_{F,M,k}$ is a Feynman propagator for the theory $I^t_0+Q_k$ obtained from the two-point function $\Delta_+$ of the state $\omega$ and where $H_{F,M,k}$ is the Hadamard parametrix of the theory $I^t_0+Q_k$.
More precisely, arguing as in the proof Proposition \ref{prop:quadratic}, we have that 
\begin{equation}\label{eq:Feynman-k}
\Delta_{F,M,k} = \sum_{n\geq 0} (-i)^n \Delta_{F,M} (\q_k\Delta_{F,M})^n 
\end{equation}
where $\Delta_{F,M}=\Delta_{+,M}+i\Delta_{A,M}$ is a Feynman propagator of the theory  $I^t_0$, which is obtained from $\Delta_F = \Delta_++i\Delta_A$ by means of the principle of perturbative agreement. 
According to Lemma 3.1 and Proposition 3.11 in \cite{Drago2015}, see also Appendix \ref{sec:perturbative-agreement} and formula \eqref{Eq: classical Moeller maps} ,
\begin{equation}\label{eq:mass-change}
\Delta_{+,M} = \mathsf{r}_M \circ \Delta_{+} \circ \mathsf{r}_M^*\,,
\end{equation}
where $\mathsf{r}_M$ is such that,  
\[
\mathsf{r}_M \chi  = \chi-\Delta_{R,M}M^{(1)}(\chi)
\]
and $\Delta_{R,M}$ is the unique retarded propagator of $I^t_0$.

We finally rewrite formula \eqref{eq:Feynman-k} in a simpler way.
We start recalling that $\Delta_{F,M} = \Delta_{+,M} +i \Delta_{A,M}$, 
hence, 
\begin{equation}\label{eq:deltafmk}
\Delta_{F,M,k} = \sum_{n\geq 0} (-i)^n (\Delta_{+,M}+i 
\Delta_{A,M}) (\q_k(\Delta_{+,M}+i \Delta_{A,M}))^n 
\end{equation}
according to Lemma 3.10 in \cite{Drago2015}, we have that 
\[
\Delta_{A,M,k} = \Delta_{A,M}\sum_{n\geq 0} (\q_k \Delta_{A,M})^n =  \Delta_{A,M}(1+\q_k \Delta_{A,M,k}) = \Delta_{A,M} \mathsf{r}_{Q_k}^*\,,
\]
hence, rearranging the sum,  we may rewrite \eqref{eq:deltafmk} as
\begin{align*}
    \Delta_{F,M,k} &=i\Delta_{A,M,k} 
    +[1+\Delta_{A,M,k}q_k] \Delta_{+,M} \mathsf{r}^*_{Q_k}  
    + \sum_{n\geq 1}   
    \mathsf{p} (\Delta_{+,M} \mathsf{r}^*_{Q_k}) (-\q_k \Delta_{+,M} \mathsf{r}^*_{Q_k})^{n}
    \\
    &= i\Delta_{A,M,k} 
    + \mathsf{r}_{Q_k} \Delta_{+,M} \mathsf{r}^*_{Q_k}  
    - \Delta_{M,k} \q_k \Delta_{+,M} \mathsf{r}^*_{Q_k}
    + \sum_{n\geq 1}   
    \mathsf{p} (\Delta_{+,M} \mathsf{r}^*_{Q_k}) (-\q_k \Delta_{+,M} \mathsf{r}^*_{Q_k})^{n} \,,
\end{align*}
where $\mathsf{p} = (1+\Delta_{A,M,k} \q_k) = 
(1+\Delta_{R,M,k} \q_k -\Delta_{M,k} \q_k) = \mathsf{r}_{Q_k} - \Delta_{M,k} \q_k$.
Notice that
\[
-\Delta_{M,k} \q_k \Delta_{+,M} = \Delta_{M,k} ((P_0+M^{(2)}-\q_k)-(P_0+M^{(2)})) \Delta_{+,M} = 0 \,,
\]
where in the last step we used the fact that $\Delta_{M,k}$ is a weak solution of $P_0+M^{(2)}-\q_k$ in both entries while $\Delta_{+,M}$ is a weak solution of $P_0+M^{(2)}$. 
Similarly we have also that for every $n\geq 2$, $(\Delta_{+,M} \mathsf{r}^*_{Q_k}) (\q_k \Delta_{+,M} \mathsf{r}^*_{Q_k})^{n-1} = 0$ because $\Delta_{+,M}\mathsf{r}^*_{Q_k} (P_0+M^{(2)}-\q_k )= 0 $. 
Finally since $\mathsf{r}_{M+Q_k}(f) = \mathsf{r}_{Q_k}(\mathsf{r}_{M}(f))$ we have that
\begin{equation}\label{eq:DeltaFk-Moller}
\Delta_{F,M,k} = i\Delta_A \mathsf{r}_{M+Q_k}^* + \mathsf{r}_{M+Q_k} \Delta_{+} \mathsf{r}_{M+Q_k}^*\,. 
\end{equation}

Combining all these observations, we conclude that the right-hand side of \eqref{approximate-Wetterich} is nothing but 
the expectation value of $\partial_k Q_k$ in a quasifree state $\omega_{M,k}$ whose two-point function is 
\begin{equation}\label{eq:two-point-function-approximation}
\Delta_{+,M,k}  =  \mathsf{r}_{M+Q_k} \Delta_{+} \mathsf{r}_{M+Q_k}^*\,,
\end{equation}
hence
\begin{equation} \label{approximate-Wetterich-2}
\partial_k U_k(\rho) = -\omega_{M,k}[\partial_k(Q_k)]\,,
\end{equation}
where $Q_k$ is a properly normal ordered Wick square.
Actually, for $q_k=k^2 f$ in the region where the cutoff function $f$ is 1, it can be evaluated as 
\[
+\partial_k U_k(\rho) =  \lim_{y\to x}  k (\Delta_{S,M,k}(y,x)-H_{M,k}(y,x))  \,,
\]
where $\Delta_{S,M,k}$ is the symmetric part of the two-point function $\Delta_{+,M,k}$ and $H_{M,k}$ is the Hadamard function related to the theory whose action is $I^t_0+Q_k=I_0+M+Q_k$.

The regularization which is provided by the point splitting procedure discussed here is compatible with the principles 
discussed in \cite{HollandsWald2001a}, see also \cite{HollandsWald14}. Furthermore, many explicit computations of similar contributions are already present in the 
literature on flat and curved spacetimes (e.g. in de Sitter \cite{Decanini2005}).

\color{black}

\begin{remark}\label{Rmk: wavefunction renormalization}

We can slightly adapt the discussion above in order to include a wavefunction renormalization. In this case, we start with an Ansatz for the average effective action in the form
\[
\Gamma_{k}(\phi) =  - \int \dd^d x  \left( \frac{z_k}{2} \nabla_a\phi \nabla^a \phi + U_k(\phi) \right) \ ,
\]
while we also modify the regulator into
\[
Q_k(\phi) = - \frac{z_k}{2} \int \dd^d x q_k(x) \phi^2(x) \ .
\]
This approximation scheme is known as LPA'. First, we rescale the fields as $\phi \to z_k^{-1/2}\phi$, to obtain
\[
\Gamma_{k}(\phi) =  - \int \dd^d x  \left( \frac{1}{2} \nabla_a \phi \nabla^a \phi + U_k \left (\frac{\phi}{\sqrt z_k} \right ) \right) \ .
\]
We can now argue as before, expanding the average effective action around a solution $\phi_{cl}$ of the quantum equations of motion for vanishing $j$, and keeping terms at most quadratic in the fluctuation field:
\[
\Gamma_k(\phi) = \Gamma_k(\phi_{cl}) 
-\int \dd^d x
\left(\frac{1}{2}
\nabla_a \varphi \nabla^a \varphi
+\frac{1}{2} U^{(2)}_k\left(\frac{\phi_{cl}}{\sqrt z_k} \right ) \varphi^2
\right)
+ \mathcal O(\varphi^3) := \Gamma^t_k(\varphi) 
+ \mathcal 
O(\varphi^3) \ ,
\]
where $U^{(2)}_k$ denotes the derivative with respect to the argument $\phi/\sqrt{z_k}$.

The truncated average effective action $\Gamma^t_k$ must come from a truncated action of the form
\[
I_0^t = - \int \dd x \left ( \frac{1}{2} \nabla_a \chi \nabla^a \chi +U^{(2)}\left( \frac{\phi_{cl}}{\sqrt z_k} \right ) \frac{\chi^2}{2} \right ) \ .
\]
Comparing the truncated action with the full action $I= I_0 + V$ we see that $I= I_0^t + \mathcal V = I_0 + M + \mathcal V$, where now $\mathcal V = V - M$ and $M$ is
\[
M = 
\int \dd x 
\bigg(m^2+\xi R-U^{(2)}_k\left( \frac{\phi_{cl}}{\sqrt z_k} \right ) \bigg) \frac{\chi^2}{2}.
\]
Scaling back to the original classical fields $\phi_{cl} \to \sqrt z_k \phi_{cl}$ we obtain
\[
M = 
\int \dd x 
\bigg(m^2+\xi R-\frac{U^{(2)}_k(\phi_{cl})}{z_k}\bigg) \frac{\chi^2}{2 }.
\]

The approximation of the r.h.s now proceeds as before, arriving at \eqref{approximate-Wetterich}. The difference now is in the mass term $M$, which includes an additional $k-$dependence in the wavefunction renormalization.

For example, in the simple case of the Minkowski vacuum as the reference state $\omega$, with $q_k = k^2$, taking the adiabatic limit and choosing a constant classical field $\phi_{cl}$, the r.h.s. of the Wetterich equation in Fourier domain becomes
\[
\partial_k \Gamma_k = \frac{1}{2(2 \pi)^d} \int \dd^d p \frac{z_k \partial_k(z_k k^2)}{z_k(p^2 + k^2 + m^2_k) + U^{(2)}_k(\phi_{cl})} \ .
\]
\end{remark}

 \color{black}

\begin{remark}\label{Rmk: local covariance in M,k}
    In the computation described above we have transformed the Wick square $T\chi^2$ in $T_{M,k}\chi^2$ with a two step procedure:
    (a) we moved $T\chi^2$ in $T_M\chi^2$ by invoking the principle of perturbative agreement;
    (b) we changed $T_M\chi^2$ into $T_{M,k}\chi^2$ roughly by computing $\omega(S(Q_k)\cdot_{T_M}T_M\chi^2)/\omega(S(Q_k))$.
    
    Although the net result is the same (i.e. we are changing the mass of the theory), the procedures (a) and (b) are slightly different.
    This difference is particularly relevant once considering the ambiguities in the choice of $T\chi^2,T_M\chi^2$ and $T_{M,k}\chi^2$.
    Indeed, since point (a) respects the principle of local covariance \cite{Brunetti2001}, it follows that the ambiguities which define $T\chi^2$ and $T_M\chi^2$ \textit{have to be the same}.
    This is not the case when discussing the ambiguities of $T_{M,k}\chi^2$, since there is no reason why (b) should respect the principle of local covariance, because theories with different $k$s are in principle not deformable one into the other.
    Thus, we are a priori free to make \textit{different} choices for the ambiguities defining $T\chi^2$ and $T_{M,k}\chi^2$.
    This possibility will turn out to be quite useful when discussing actual computation, \textit{cf.} Section \ref{Subsubsec: Vacuum case}.
\end{remark}

Taking successive functional derivatives with respect to $\phi$ on both sides of \eqref{approximate-Wetterich} gives the beta-functions for the evolution equations of the running coupling constants, defined as the coefficients of the Taylor series of $U_k(\phi)$ close to $\phi=0$ (which correspond to a local minimum of $U_k(\phi)$):
\begin{equation} \label{effective-potential-taylor}
U_k(\phi) = U_{0,k} + m^2_k \frac{\phi^2}{2} +  \lambda_k\frac{\phi^4}{4!} + \mathcal O(\phi^6) \ .
\end{equation}
In the previous discussion, we always tacitly assumed that the starting point for the flow was the classical action $I_0$.
Clearly, one may start the flow at a different scale $\Lambda$, where $\Lambda\to\infty$ corresponds to the classical action $I_0$.
Such identification lets one interpret the beta-functions as the evolution equations for the couplings.

The beta-functions are then defined as the evolution equations for the dimensionless parameters $\tilde m^2_k$ and $\tilde \lambda_k$ with respect to the renormalization time $t = \log k/ \Lambda$, where $\Lambda$ is the scale at which the renormalization starts. In powers of $k$, the dimensions of the couplings are $[m^2_{\beta, k}] = 2$ and $[\lambda_k] = 4 - d$. Then we have
\begin{align}
k \partial_k \tilde m^2_{k,\beta} &= k^{-1} \partial_k m^2_{k,\beta} - 2 \tilde m^2_{k,\beta}  \\
k \partial_k \tilde \lambda_k &= k^{d-3} \partial_k \lambda_k + (d-4) \tilde \lambda_k \ . 
\end{align}

\section{Applications} \label{section-applications}

\subsection{Local regulator and the high-temperature fixed point in $\lambda \phi^4$}
    
As an application of the theoretical machinery presented in this paper, we apply our modified Wetterich equation to $\lambda \phi^4$ model at finite inverse temperature $\beta$.
    The Lagrangian density for $\lambda \phi^4$ at finite temperature is given by
\[
\mathcal L =  - \frac{1}{2} \nabla_a \chi \nabla^a \chi - \frac{\lambda}{4!}\chi^4 \ .
\]

We apply the renormalization scheme of section \ref{sec:LPA}. We therefore linearize the theory as a free theory for the perturbation $\varphi$ with a mass term $m^2 + \lambda \rho$ where $\rho=\phi^2/2$. Following \eqref{effective-potential-taylor}, our Ansatz for $U_k({\phi})$ is
\begin{equation} \label{explicit-effective-potential}
U_k({\phi}) = U_{0,k,\beta} + m^2_{k, \beta} \rho +\frac{1}{6}\lambda_{k, \beta}\rho^2 \,\qquad \rho=\frac{\phi^2}{2} .
\end{equation}
The couplings $m^2_{\beta,k}$ and $\lambda_{\beta,k}$ can depend on the temperature as well as $k$, since, for example, there will be contributions coming from the one-loop renormalization of the thermal mass.

The right-hand side of the modified Wetterich equation in the LPA \eqref{approximate-Wetterich} can be computed as follows. 
We choose as reference state an equilibrium state with respect to Minkowski time evolution at inverse temperature $\beta$, which is also quasifree for the free theory. This is a KMS state for the free theory $\lambda = 0$ whose two-point function is invariant under translations, and has the form (see e.g. \cite{Drago2015})
\begin{equation}
\Delta_{+}^\beta(x^0,\mathbf{x}; y^0,\mathbf{y}) = \int \frac{\dd^{d-1} \mathbf{p}}{(2 \pi)^{d-1}} e^{i\mathbf{p}\cdot(\mathbf{x} - \mathbf{y})} \frac{1}{2 w } \bigg(\frac{e^{-i  w (x^0-y^0)}}{1-e^{-\beta w}} + \frac{e^{i w (x^0-y^0)}}{e^{\beta w} -1} \bigg) \ ,
\end{equation}
where $w = |\mathbf{p}|$ . 
We follow the procedure discussed above to get the corresponding  
\[
\Delta^\beta_{+,M,k}=
\mathsf{r}_{Q_k+M}\Delta^\beta_{+} \mathsf{r}_{Q_k+M}^*.
\]
Following \cite{Drago2015}, see also \cite{DragoGerard2016,Drago2018,Vasconcellos2019}, and taking the adiabatic limit we have that the two-point function of the state in which we compute the expectation value of the Wick square is
\begin{equation}\label{eq:2pt-beta-k-M}
\Delta_{+,M,k}^\beta(x^0,\mathbf{x}; y^0,\mathbf{y}) = \int \frac{\dd^{d-1} \mathbf{p}}{(2 \pi)^{d-1}} e^{i\mathbf{p}\cdot(\mathbf{x} - \mathbf{y})} \frac{1}{2w_{M,k} } \bigg(\frac{e^{-i  w_{M,k} (x^0-y^0)}}{1-e^{-\beta w}} + \frac{e^{i w_{M,k} (x^0-y^0)}}{e^{\beta w} -1} \bigg) \ ,
\end{equation}
where $w_{M,k} = \sqrt{w^2 + m_{k, \beta}^2 + \lambda_{k, \beta} \rho +\q_k}$.
Notice that the $w$-factors associated with the modes have changed, while this is not the case for the Bose factors.

The above two-point function differs from that of \cite{Litim2006}, since in our construction the Bose factors are $k$-independent. The reason is that, in the derivation of the Wetterich equation we have presented in the previous sections, we fixed the state once and for all in the original, unregularised theory. 

One could ask what happens if, instead, one chooses different  states at different $k$'s. 
 In the example of thermal states a possibility in this direction 
 would be to consider $\beta$ which depends on $k$.
However, 
this introduces 
a new scale dependence through the choice of 
states at various $k$
in addition to the the scale dependence one introduces in the observables. As a consequence, the flow equation \eqref{approximate-Wetterich} 
would receive an additional contribution from the derivative of the explicit dependence of the state on $k$ and the obtained equation will not be of the simple form given in \eqref{eq:Wetterich-Local}.
With this in mind, we conclude that the flow equation \eqref{approximate-Wetterich}
takes the form
\begin{equation} \label{eq:wetterich-termal}
\partial_k U_k({\phi}) =  
\lim_{y\to x} 
(\Delta^\infty_{S,M,k}(y,x)-H(y,x))  \frac{\partial_k \q_k(x)}{2}
+
\lim_{y\to x}
(\Delta^\beta_{S,M,k}(y,x)-\Delta^\infty_{S,M,k}(y,x))  \frac{\partial_k \q_k(x)}{2}
\end{equation}
The first contribution is the one which would remain in the limit $\beta\to\infty$, namely when $\omega$ is chosen to be the vacuum state, while the second is the correction due to the temperature. 

The flow equations for $m_{k,\beta}$ and $\lambda_{k,\beta}$ are obtained from the above, taking on both sides of the equation functional derivatives up to order two with respect to $\rho$ and equating them for $\rho=0$.
Hence, to get the beta-functions for both $\lambda$ and $m$ we analyse these two contributions separately.
We have that (in the adiabatic limit and in the four dimensional case)
\begin{equation}\label{Eq: thermal-minus-vacuum contribution}
\begin{aligned}
    A&:=
    \lim_{y\to x}
    (\Delta^\beta_{S,M,k}(y,x)-\Delta^\infty_{S,M,k}(y,x)(y,x))  \frac{\partial_k \q_k(x)}{2}
    \\&= \frac{1}{(2\pi)^{3}}  \int \dd^{3} \textbf{p} 
     \frac{1}{w_{M,k} } \bigg(\frac{1}{e^{\beta w} -1} \bigg) \frac{\partial_k \q_k(x)}{2}
\end{aligned}
\end{equation}
while, by employing the result recalled in Appendix \ref{Sec: Hadamard expansion of the Minkowski vacuum two-point function and the Wick square},
\begin{equation}
\label{Eq: vacuum contribution}
\begin{aligned}
    B&:=
    \lim_{y\to x}
    (\Delta^\infty_{S,M,k}(y,x)-H(y,x)) 
    \frac{\partial_k \q_k(x)}{2}
    \\&= 
    \frac{1}{8\pi^2}(k^2 + m_{k,\beta}^2 + \lambda_{k,\beta}\rho)
    \log \left( \frac{k^2 + m_{k,\beta}^2 + \lambda_{k,\beta}\rho}{\mu^2} 
    \right) 
    \frac{\partial_k \q_k(x)}{2}
\end{aligned}
\end{equation}
the regulator $\q_k$ is now chosen to be equal to $\q_k=k^2$ ---the corresponding adiabatic limit has been tacitly taken.

\subsubsection{Vacuum case}\label{Subsubsec: Vacuum case}

In the limit of vanishing temperature $\beta\to\infty$ we have that the contributions due to \eqref{Eq: thermal-minus-vacuum contribution} vanish, and in the four dimensional case we are left with
\begin{align*}
k\partial_k U_{0,k} &= \frac{1}{8\pi^2}  k^2 (k^2+m_k^2) \log \left(\frac{k^2+m_k^2}{\mu^2} \right)
\\
k\partial_k m^2_k &= 
\frac{1}{8\pi^2}k^2
\left(1+ \log \left(\frac{k^2+m_k^2}{\mu^2} \right)
\right) \lambda_k
\\
k\partial_k \lambda_k &= 
\frac{3}{8\pi^2} 
\frac{ k^2}{k^2+m_k^2}\lambda_k^2 \ .
\end{align*}
In terms of the dimensionless couplings, $\tilde{U}_{0,k} = U_{0,k}/k^4$ $\tilde{m}_k = m_k/k$ and $\tilde{\lambda}_k=\lambda_k$
the beta-functions then are
\begin{align*}
k\partial_k \tilde{U}_{0,k} &= -4\tilde{U}_{0,k}+\frac{1}{8\pi^2}  (1+\tilde{m}_k^2) \left[\log \left({1+\tilde{m}_k^2} \right) +
  \log \left(\frac{k^2}{\mu^2} \right) 
\right]
\\
k \partial_k \tilde m^2_k &= -2 \tilde m^2_k + \frac{1}{8 \pi^2}   \left [ 1 +   \log(1+ \tilde m^2_k) + \log(\frac{k^2}{\mu^2})  \right ] \tilde \lambda_k \\
k \partial_k \tilde \lambda_k &=  \frac{3 }{8\pi^2 }  \frac{\tilde \lambda^2_k}{ 1+\tilde m_k^2 } 
\end{align*}
 In the above equation, the arbitrary mass parameter $\mu$ represents a residual freedom in the ultraviolet renormalization scheme we have adopted. Hence, the price to pay to have a local regularization term, is an additional freedom in the beta-functions due to the ultraviolet renormalization scale $\mu$. However, the additional freedom may be safely removed setting
$\mu=k$. 
This is equivalent to tune the renormalization ambiguities of $T_{M,k}\chi^2$ which, we recall, are not forced to be the same as the ones present in $T\chi^2$ ---\textit{cf.} Remark \ref{Rmk: local covariance in M,k}.
Notice that this choice would not change the form of the Wetterich equation as it is made \textit{after} deriving in $k$.
We finally observe that setting $\mu=k$ we can still identify a fixed point for $\tilde\lambda_k=0$ and for $\tilde{m}_k=0$, in agreement with what one obtains using non local regulators, showing that the only fixed point in four dimensions is the non-interacting one.

\subsubsection{High temperature limit}

We now take the high temperature limit $\beta \to 0$.
We shall later approximate $\dfrac{1}{e^{\beta w}-1}\approx \dfrac{1}{\beta w}$ in some part of the computation.
Hence 
\begin{align*}
A 
=      
 \frac{k^3}{2\pi^{2}}  \int_0^\infty \dd p   
 \frac{p^2}{\sqrt{p^2 + \left(\frac{m_{k,\beta}}{k}\right)^2 +\frac{\lambda_{k,\beta}}{k^2}\rho + 1}}
 \frac{1}{e^{\beta k p}-1}\,.
\end{align*}
Expanding $A$ up to order $2$ in powers of $\rho$ we get
\begin{multline*}
A 
\simeq  
 \frac{k^3}{4\pi^{2}}  \int_0^\infty \dd p^2 
 \left[ 
 \frac{1}{\left({p^2 + \left(\frac{m_{k,\beta}}{k}\right)^2  + 1}\right)^{\frac{1}{2}}}
 -
 \frac{1}{2 k^2}
 \frac{\lambda_{k,\beta}\rho}{\left({p^2 + \left(\frac{m_{k,\beta}}{k}\right)^2  + 1}\right)^{\frac{3}{2}}}
 \right.
 \\\left.+
 \frac{3}{8 k^4}
 \frac{(\lambda_{k,\beta}\rho)^2}{\left({p^2 + \left(\frac{m_{k,\beta}}{k}\right)^2  + 1}\right)^{\frac{5}{2}}}
  \right ]   
  \frac{p}{e^{\beta k p}-1}\,.
\end{multline*}
The contribution to the beta-functions due to $A$ diverges as $1/\beta$ in the limit $\beta\to 0$ while $B$ stays bounded ---\textit{cf.} Equation \eqref{Eq: thermal-minus-vacuum contribution}-\eqref{Eq: vacuum contribution}. Thus $A$ dominates over $B$, and we can drop the latter term when computing the flow equations for the coupling parameters in the high-temperature regime.
We then find
\begin{align*}
k \partial_k \tilde{U}_{0, k, \beta} &= - \tilde{U}_{0, k, \beta} 
+
 \frac{\zeta(3)}{2\pi^{2}}
 \\
k \partial_k \tilde{m}_{k, \beta}^2 &= 
-2(\tilde{m}_{k, \beta})^2 
- \frac{1}{2\pi^{2}}
 \frac{\tilde \lambda_{k,\beta}}{\left( 1+ \tilde{m}_{k,\beta}^2 \right)^{\frac{1}{2}}}
\\
k \partial_k \tilde{\lambda}_{k, \beta} &=  
- \tilde{\lambda}_{k, \beta}
+
 \frac{3}{8\pi^{2}}
 \frac{(\tilde \lambda_{k,\beta})^2}{\left( 1+ \tilde{m}_{k,\beta}^2 \right)^{\frac{3}{2}}} \, ,
\end{align*}
where  $\zeta$ is the Riemann zeta function, and we introduced the dimensionless, rescaled constants
$\tilde{U}_{0,k,\beta}={U}_{0,k,\beta} \beta^2 /k$, 
 $\tilde{m}_{k,\beta} = m_{k,\beta}/k$ and $\tilde{\lambda}_{k,\beta} = \lambda_{k,\beta}/(\beta k)$.
 In $d=4$, one finds a non-trivial fixed point for 
$\tilde{U}_* = \zeta(3)/2\pi^{2}$, 
$\tilde m^2_* = -2/5$, $ \tilde{\lambda}_* =  (8/3)\pi^2(1+ {\tilde m^2_*})^{3/2}$. The minus sign in the mass fixed point indicates that the symmetry $\chi\to-\chi$ is spontaneously broken in the chosen state.

To make it clearer, following \cite{Tetradis1992}, we can now repeat the above analysis, in which the effective potential takes the simple form
\[
U_{k} = \frac{\lambda_{ k, \beta}}{2}(\rho - \rho_{0, k, \beta})^2 \ \qquad {\rho=\frac{\phi^2}{2}}.
\]
which coincides with the effective potential written in equation \eqref{explicit-effective-potential} up to a constant.
The new parameter $\rho_{0, k, \beta}=\phi^2_{0, k, \beta}/2$ is the minimum of the potential, located at $U^{(1)}_{ k}({\phi_{0, k, \beta}}) = 0$. The new parameters are then linked to the old coupling constants, in particular
$m^2_{ k, \beta} = - \lambda_{ k, \beta} \rho_{0, k, \beta}$, while $\lambda_{ k, \beta}$ is scaled by a factor 3.
The flow equation for these parameters can be obtained from those written above.

\subsubsection{de Sitter space}
As a final application we show that the same methods apply in the context of curved spacetimes; in particular we study (the linearization of) $\lambda \chi^4$ in de Sitter space, defined as the four dimensional hyperboloid embedded in five dimensional flat space via the equation
\[
X^a X^b\eta_{ab} = H^{-2} \ , 
\]
where $H > 0$ is the Hubble constant and $\eta$ is the five-dimensional Minkowskian metric. 
We consider the linear theory to be in the Bunch Davies state \cite{Bunch1978}, which is the unique quasifree maximally symmetric state on the de Sitter spacetime. Following \cite{Birrell1984, Allen1987}, the symmetric part of its two-point function is
\[
\Delta_S^{\text{BD},+}(x,y) = 
\frac{H^2}{16\pi }\frac{\left(\frac{1}{4}-\nu^2\right)}{\cos(\pi\nu)} 
{_2F_1}\left(\frac{3}{2}+\nu,\frac{3}{2}-\nu;2;\frac{1+Z(x,y)}{2}\right)
\] 
where
$_2F_1$ is the hypergeometric function, $Z(x,y)=H^2 X^a(x) X^b(y)\eta_{ab}$
is related to the geodesic distance $d(x, y)=H\cos(Z(x,y))$ between $x$ and $y$, and 
\begin{equation}\label{eq:nu}
\nu = \sqrt{\frac{9}{4} - 12\xi + \frac{m^2}{H^2}}\,,
\end{equation}
where $m$ is the mass of the quantum field and $\xi$ its coupling to the scalar curvature.

We proceed with the following Ansatz
\begin{equation*}
U_k({\phi}) = m^2_k \rho +\frac{1}{6}\lambda_k\rho^2 , \qquad {\rho=\frac{\phi^2}{2}} \ .
\end{equation*} 
Notice that, for the sake of simplicity, $\xi$ does not depend on $k$.

To apply the approximation scheme introduced above, we need first of all to apply the maps which realise the classical transformation  
\[
\mathsf{r}_{M+Q_k}\Delta^{\text{BD}}_S\mathsf{r}_{M+Q_k}^*.
\]
However, instead of directly performing that computation, we observe that 
in the adiabatic limit the obtained states share the same symmetry properties as those of the original two-point function, because the classical M\o ller map preserves the spacetime symmetry.
The only maximally symmetric state in de Sitter is the Bunch Davies state and the original state is maximally symmetric; hence, the new state needs to be a Bunch Davies state too with a mass $m^2 = m_k^2 +\lambda_k \rho +k^2$.

For massive theories, with a general, non-minimal coupling $\xi$, the renormalized expectation value (via the Hadamard procedure) of the Wick square $\chi^2$ in this state is given by \cite{Bernard1986}
\begin{multline}
\omega(\chi^2(x)) = \\ -\frac{1}{16 \pi}\bigg \{ - \frac{2H^2}{3} + \bigg [ (m_k^2 + k^2 + \lambda_k \rho) + \bigg(\xi - \frac{1}{6} \bigg )12 H^2 \bigg ]  \bigg[\psi\bigg(\frac{3}{2} + \nu\bigg) + \psi\bigg(\frac{3}{2} - \nu\bigg) + \log\bigg(\frac{12 H^2}{\mu^2}\bigg) \bigg] \bigg \} \ ,
\end{multline}
where $\psi$ is the digamma function, defined as the logarithmic derivative of the Euler gamma function, $\nu$ is as in \eqref{eq:nu} with mass square equal to $m^2+k^2 + \lambda_k \rho$, and $\mu$ is again an arbitrary mass parameter.

Starting from this expression for  $\omega(\chi^2)$, the evolution equations become
\begin{multline*}
k\partial_k m^2_k 
= - \frac{k^2 \lambda_k}{16 \pi^2}
\bigg \{ \log \frac{12 H^2}{\mu^2} + \psi(\frac{3}{2}-\nu) + \psi(\frac{3}{2}+\nu) 
\\ + \frac{1}{2\nu} \bigg[\frac{m^2_k}{H^2}+ \frac{k^2}{H^2} + 12(\xi - \frac{1}{6} )\bigg ] [\psi'(\frac{3}{2}-\nu) - \psi'(\frac{3}{2}+\nu) ] \bigg) \bigg \}
\end{multline*}


\begin{multline*}
  k\partial_k \lambda_k =
  - \frac{3 k^2 \lambda_k^2}{16 \pi^2 H^2 \nu} \bigg \{ \psi'\left(\frac{3}{2}- \nu \right) - \psi'\left(\frac{3}{2} + \nu \right) + \\
  \frac{k^2+m^2 + 12 H^2\left (\xi - \frac{1}{6}\right )}{4H^2 \nu^2} \left [ \psi'\left(\frac{3}{2}- \nu \right) - \psi'\left(\frac{3}{2}+ \nu \right) + \nu \left ( \psi''\left(\frac{3}{2}- \nu  \right) + \psi''\left(\frac{3}{2} + \nu \right)\right )\right ]
  \bigg \} \ .
\end{multline*}
We can now define new, dimensionless couplings
\[
\tilde m^2_k := \frac{m^2_k}{H^2} \ , \quad \tilde \lambda_k = \frac{k^2}{H^2} \lambda_k \ .
\]
Notice that, as in the thermal case, the appearance of a dimensionful parameter ($H$ in this case) allows for a different scaling behaviour of the coupling constants.

In terms of the rescaled couplings the beta-functions become
\begin{align*}
   k \partial_k \tilde m^2_k &= - \frac{\tilde \lambda_k}{16 \pi^2}
\bigg \{ \log \frac{12 H^2}{\mu^2} + \psi(\frac{3}{2}-\nu) + \psi(\frac{3}{2}+\nu) 
\\ 
&+ \frac{1}{2\nu} \bigg[\frac{k^2}{H^2} + \tilde m^2_k + 12(\xi - \frac{1}{6} )\bigg ] [\psi'(\frac{3}{2}-\nu) - \psi'(\frac{3}{2}+\nu) ] \bigg) \bigg \}
\end{align*}
\begin{multline*}
    k \partial_k \tilde \lambda_k = 2 \tilde \lambda_k
     - \frac{3 \tilde \lambda_k^2}{16 \pi^2 \nu} \bigg \{ \psi'\left(\frac{3}{2}- \nu \right) - \psi'\left(\frac{3}{2} + \nu \right) \\
 + \frac{1}{4\nu^2} \left(\frac{k^2}{H^2}+\tilde m_k^2 + 12 \left (\xi - \frac{1}{6}\right ) \right ) \bigg [ \psi'\left(\frac{3}{2}- \nu \right) - \psi'\left(\frac{3}{2}+ \nu \right) \\
 +\nu \left ( \psi''\left(\frac{3}{2}- \nu  \right) + \psi''\left(\frac{3}{2} + \nu \right)\right )\bigg ]
  \bigg \} \ .
\end{multline*}
By choosing $\mu^2 = 12 H^2$, we can remove all the dependence on the additional parameter $\mu$. This is possible in de Sitter since we have a new mass scale $H^{-2}$ which enters the flow equations as an "external parameter", and does not depend on the scale, similar to the inverse temperature $\beta$ in the flow equation for thermal theories we considered in the last section. Comparing with the Minkowski equations, we see that the first term in the beta-function for $\tilde \lambda_k$, corresponds to an effective dimension of $2$ in the flow equation for $\tilde\lambda_k$.

In the limit $k^2/H^2 \to 0$, corresponding to an inflationary regime in de Sitter, the RG flow equations acquire a non-trivial fixed point $(\tilde m^2_*, \tilde \lambda_*)$.

Due to its complexity, we do not report here the full expression for general $\xi$. Choosing the conformal coupling $\xi = 1/6$, the non-trivial fixed point is given by the simple expression
\begin{equation}
    \tilde m^2_k \simeq 0.164588 \ , \quad \tilde \lambda_k \simeq 179.237
\end{equation}

\section{Conclusions and Outlook}
In this paper, we employed the formalism of pAQFT to generalise the methods of the FRG to globally hyperbolic spacetimes and for generic states of a scalar field. We then showed, with the examples of the Minkowski vacuum, thermal states, and the Bunch-Davies vacuum in de Sitter, how to perform practical computations out of our formalism, that can be compared with the known results in the literature.

From here, several interesting directions of research open. On the conceptual side, we notice how the entire construction is based on the definition of the generating functional $Z_k(j)$ as the expectation value of a suitably regularised relative S-matrix. It would be interesting, then, to make the connection with the $C^*-$algebraic approach to interacting QFT first developed in \cite{Buchholz2019}, where primary objects are axiomatically defined relative S-matrices. Such a method would go beyond the perturbative construction, permitting an exact reformulation of the FRG equations.

The pAQFT approach already encompasses gauge theories \cite{Hollands2007} and perturbative quantum gravity \cite{Brunetti2016} via the BV formalism \cite{FR, Fredenhagen2013}. The generalization of our formalism to gauge theories seems in reach, and it would provide a clear conceptual framework and a rigorous construction to e.g. the thermal effects in non-abelian gauge theories, playing an important role in the phases of QCD (see section 5 in \cite{Dupuis2021} and references therein, with first developments in \cite{Reuter1993, Reuter1994, Gies2002}). 

On the other hand, the generalization to quantum gravity would provide mathematical tools to rigorously explore the UV properties of Lorentzian quantum gravity, making contact with the rich literature of the asymptotic safety community \cite{Reuter1996, Eichhorn2019}. In particular, the pAQFT formulation of quantum gravity is naturally developed in Lorentzian spacetimes, and it would be interesting to compare it with the results obtained in Asymptotically Safe Lorentzian Quantum Gravity \cite{Fehre2021, Manrique2011}. pAQFT could in principle address long-standing questions in the asymptotic safety program \cite{Donoghue2020, Bonanno2020}: first of all, pAQFT could provides a natural Lorentzian framework, without the need of a Wick rotation to Euclidean signature; secondly,  there are already developments on the use of relational observables \cite{Brunetti2016, Baldazzi2021} which could be compared with cosmological observations. In particular, the role of the state in the flow equations, as emphasised in our reformulation of the Wetterich equation \eqref{eq:Wetterich-Local}, could provide a new perspective on the objection raised by Donoghue \cite{Donoghue2020}, as the different possible runnings of the Newton's constant depending on the scattering process considered in effective field theory.

Finally, from a more practical point of view, our formalism provides a natural language to discuss in a systematic way the flow equations in generic states, as thermal states, or on curved spacetimes, where a natural notion of a vacuum is in general not at our disposal. It would be interesting, for example, to apply the formalism in cosmological situations, as in de Sitter space, and compare the results with the existing literature \cite{Serreau2011, Serreau2014, Guilleux2015, Guilleux2016}, or on black hole spacetimes, where self-interaction effects could provide instabilities of the theory.

\paragraph{Acknowledgements.}
We are grateful to Rudi Banerjee, Astrid Eichhorn and Max Niedermaier for interesting comments on this manuscript and also for pointing out to us some relevant references in the asymptotic safety literature. We thank both referees far useful comments and suggestions on an earlier version of this paper.
K.R. found the discussions with Astrid Eichhorn, Benjamin Knorr, Alessia Platania and Frank Saueressig very helpful and inspiring.
The work of E.D. is supported by a PhD scholarship of the University of Genoa.
E.D., N.D. and N.P. are grateful for the support of the National Group of Mathematical Physics (GNFM-INdAM).

\paragraph{Data availability statement.}
Data sharing is not applicable to this article as no new data were created or analysed in this study.

\appendix

\section{Principle of Perturbative Agreement}\label{Sec: Principle of Perturbative Agreement}\label{sec:perturbative-agreement}
	In this appendix we recall some basic facts about the principle of perturbative agreement. For further reading we refer to \cite{HollandsWald2004} and to \cite{Drago2015,Zahn2015}.

    Let $I$ be the action $I:=I_0+V$ where $I_0$ and $V$ are as in \eqref{eq:full-action}.
    Moreover, let $Q_k$ the quadratic potential depending on a parameter $k$
    \begin{align}\label{Eq: Q-functional}
    	Q_k(\chi)= -\frac{1}{2}\int \dd^d x \q_k(x) \chi^2(x) ,
    \end{align}
	where $q_k=k^2f\in C_{\mathrm{c}}^\infty(\mathcal{M})$ and $f\in C_{\mathrm{c}}^\infty(\mathcal{M})$ is the usual cutoff which can be eventually removed by performing a limit $f\to 1$ in a suitable sense.
	In what follows we consider
    \begin{align*}
    	I_{k}
    	:=I+Q_k
    	=I_0+(V+Q_k)
    	=I_{0k}+V
    	=I_0+V_k\,,
    \end{align*}
and we denote with $\star$, $\cdot_T$, $\mathcal{A}$, $S$, $R$ (\textit{resp.} $\star_k$, $\cdot_{T_{k}}$, $\mathcal{A}_k$, $S_k$, $R_k$) the star product, time-ordered product, algebra of observables, $S$-matrix and M\o ller map associated with $I_0$ (\textit{resp.} $I_{0k}$), respectively.

	Considering the (abstract, unreachable) interacting algebra $\mathcal{A}_{I_k}$, we have two maps which defines a perturbative representation of $\mathcal{A}_{I_k}$ in either $\mathcal{A}[[V_k]]$ or $\mathcal{A}_k[[V]]$ ---depending on whether $Q_k$ is considered as part of $I_0$ or of $V$.
	These are the quantum M\o ller maps \eqref{eq:Bogoliubov} $R_{V_k}$ and $R_{k,V}$. 	
	which permit to represent the generators of the interacting algebra, \textit{i.e} the local interacting fields, to $\mathcal{A}[[V_k]]$ or in $\mathcal{A}_k[[V]]$.	
	Notice that the M\o ller maps $\RQ_{V_k}$ (\textit{resp.} $\RQ_{k,V}$) are constructed with the time ordered exponential of $TV_k$ (\textit{resp}. $T_k V$) in order to make the above representation fully local and covariant.
	Notably, the two representations are related as follows.
	First there exists a (classical M\o ller) isomorphism $\RC_{Q_k}\colon\mathcal{A}_k\to\mathcal{A}$ defined by
	\begin{align}\label{Eq: classical Moeller maps}
		\RC_{Q_k}\colon\mathcal{A}_k\to\mathcal{A}\,,\qquad
		(\RC_{Q_k}F)(\chi)=F(\mathsf{r}_{Q_k}\chi)\,,\qquad
		\mathsf{r}_{Q_k}\chi:=\chi-\Delta_{R,k}q_k\chi\,
	\end{align}
	where $\Delta_{R,k}$ is the retarded operator associated to $I_{0k}$.
	Notice in particular that $r_{Q_k}$ intertwines between $I_{0k}$ and $I_0$ namely $I_{0k}^{(1)}r_{Q_k}=I_0^{(1)}$.
	(Notice that the existence of $\Delta_{R,k}$ is ensured by the fact that $Q_k$ is local and does not contains second derivatives.)
	As a matter of fact not only $\mathcal{A}\simeq\mathcal{A}_{k}$, but also the time-ordered products can be related.
	In particular, using $\tilde{\Upsilon}$ given in \eqref{eq:alpha}
	\begin{align}\label{Eq: PPA map}
		\gamma_k\colon\mathcal{F}_{\text{loc}}\to\mathcal{F}_{\text{loc}}\,,\qquad
		\gamma_k F:=e^{\tilde{\Upsilon}_{\Delta_{F,k}-\Delta_F}}F\,,\qquad
	\end{align}
	is such that
	\begin{align}\label{Eq: the PPA map intertwines the T-products}
		\gamma_k(F\cdot_TG)=\gamma_k(F)\cdot_{T_k}\gamma_k(G)\,.
	\end{align}
	\begin{remark}\label{Rmk: T, TLambda Wick ordering}
		The map $\gamma_k$ can also be used to define a Wick-ordering map for $\mathcal{A}_k$ \cite{Drago2015,HollandsWald2004}.
		In particular, given the Wick ordering map $T\colon\mathcal{F}_{\text{loc}}\to\mathcal{F}_{\text{loc}}$ for the algebra $\mathcal{A}$, we consider $\gamma_k\circ T$.
		At this stage one may prove, \textit{cf.} \cite{HollandsWald2004}, that the renormalization ambiguities of the Wick ordering map $T$ can be chosen so that $T_k=\gamma_k\circ T\colon\mathcal{F}_{\textrm{loc}}\to\mathcal{F}_{\textrm{loc}}$ is a Wick ordering map for $\mathcal{A}_k$.
	\end{remark}
	It is not difficult to show that, in a perturbative sense,
	\begin{align}
		\gamma_k={\RC}_{Q_k}^{-1}\circ \RQ_{Q_k}\,,
	\end{align}
	where $\RQ_{Q_k}\colon\mathcal{A}_k\to\mathcal{A}[[Q_k]]$ is the (perturbative) representation of $\mathcal{A}_k$ in $\mathcal{A}[[Q_k]]$.
	
	The latter relation can be promoted to the interacting setting ---\textit{i.e.} $V\neq 0$.
	In particular the following  holds, \textit{cf.} \cite{Drago2015}:
	
	\begin{lemma}
	It holds that 
	\begin{align}\label{Eq: gPPA for Moeller operators}
		\RQ_{V_k}
		=\RC_{Q_k}
		\circ \RQ_{k,{\gamma_k} V}
		\circ\gamma_k\,,
	\end{align}
	and similarly
	\begin{align}\label{Eq: gPPA for relative S-matrices}
		S_{V_k}(F)=\RC_{Q_k}S_{k,{\gamma_k} V}(\gamma_k F)\,.
	\end{align}
    Furthermore, $\gamma_k V$ and $V$ differs only by a different choice of renormalization constants.
		\end{lemma}
	\begin{proof} 
To prove Eq. \eqref{Eq: gPPA for Moeller operators} and \eqref{Eq: gPPA for relative S-matrices} we proceed by direct inspection
		\begin{align*}
			\RC_{Q_k}S_{k,{\gamma_k} V}(\gamma_k F)	&=\RC_{Q_k}\left[S_k({\gamma_k}V)^{-1}\star_k S_k(\gamma_k(F+V))
			\right]
			\\&=\left[\RC_{Q_k}S_k( {\gamma_k}V)\right]^{-1}
			\star \RC_{Q_k}\circ\gamma_k S(F+V)
			\\&=\left[\RC_{Q_k}S_k({\gamma_k}V)\right]^{-1}
			\star \RQ_{Q_k} S(F+V)
			\\&=\left[\RC_{Q_k}S_k({\gamma_k} V)\right]^{-1}
			\star S(Q_k)^{-1}\star S(F+V_k)\,.
		\end{align*}
		Moreover
		\begin{align*}
			\RC_{Q_k}S_k(V)
			=\RC_{Q_k}\gamma_k S(V)
			={\RQ}_{Q_k} S(V)
			=S(Q_k)^{-1}\star S(V_k)\,,
		\end{align*}
		so that overall we have
		\begin{align*}
			\RC_{Q_k}S_{k,{\gamma_k} V}(\gamma_k F)
			&=\left[\RC_{Q_k}S_k({\gamma_k} V)\right]^{-1}
			\star S(Q_k)^{-1}\star S(F+V_k)\,.
			\\&=\left[S(Q_k)^{-1}\star S(V_k)\right]^{-1}
			\star S(Q_k)^{-1}\star S(F+V_k)
			=S_{V_k}(F)\,.
		\end{align*}
		This proves Equation \eqref{Eq: gPPA for relative S-matrices}.
		Concerning Equation \eqref{Eq: gPPA for Moeller operators} this follows by Equation \eqref{Eq: gPPA for relative S-matrices} as
		\begin{align*}
			\RQ_{V_k} F
			=\frac{\mathrm{d}}{i\mathrm{d}\mu}S_{V_k}(\mu F)\bigg|_{\mu=0}
			\stackrel{\eqref{Eq: gPPA for relative S-matrices}}{=}\frac{\mathrm{d}}{i\mathrm{d}\mu}\RC_{Q_k}S_{k,{\gamma_k}V}(\mu\gamma_k F)\bigg|_{\mu=0}
			=\RC_{Q_k}\RQ_{k,{\gamma_k}V} \gamma_k F\,.
		\end{align*}
	\end{proof}

\section{Technical lemmata}\label{App: technical lemmata}

We collect in this appendix some technical lemmata used in the main text.
The first is used to evaluate the effect of the product of the time ordered exponential of local currents with a local field.
\begin{lemma}\label{Lem: time-ordered multiplication with S(J)}
	For all $F\in\mathcal{A}$ we have
	\begin{align}\label{Eq: time-ordered multiplication with S(J)}
		[S(J)\cdot_T F](\chi)
		=e^{-\Delta_F(j,j)/2}e^{iJ(\chi)}F(\chi+i\Delta_F j)\,.
	\end{align}
\end{lemma}
\begin{proof}
	By direct inspection we have
	\begin{align*}
		S(J)=e^{iJ}_{\cdot_T}
		=T(e^{iJ}_\otimes)
		=\sum_{n\geq 0}\frac{1}{n!2^n}\Delta_F^{\otimes n}(ij)^{\otimes 2n} e^{iJ}
		=e^{-\Delta_F(j,j)/2}e^{iJ}\,.
	\end{align*}
	Moreover for all $F\in\mathcal{A}$ we have
	\begin{align*}
		(e^{iJ}\cdot_T F)(\chi)
		=\sum_{n\geq 0}\frac{1}{n!}\Delta_F^{\otimes n}(ij)^{\otimes n} e^{iJ(\chi)} F^{(n)}(\chi)
		=e^{iJ(\chi)}F(\chi+i\Delta_F j)\,.
	\end{align*}
	Combining these results leads to equation \eqref{Eq: time-ordered multiplication with S(J)}.
\end{proof}
\begin{remark}\label{Rmk: time-ordered multiplication with S(J) - special case}
	In the particular case of $j=-\mathcal{L}_0^{(1)}\phi$ equation \eqref{Eq: time-ordered multiplication with S(J)} simplifies to
	\begin{align*}
		[S(J)\cdot_T F](\chi)
		=e^{-\mathcal{L}_0(\phi)}e^{iJ(\chi)}F(\chi{+}\phi)\,,
	\end{align*}
	where we used that $\Delta_F\mathcal{L}_0^{(1)}=i\delta$.
\end{remark}

In the next we analyse the vanishing of the various propagators of a theory whose action is $I_0+Q_k$ in the limit $k\to\infty$. 
\begin{lemma}\label{le:limit-propagators}
Let $\q_k(x)=  k^2 f(x)$ where the cutoff function $f\in C^{\infty}_{\mathrm{c}}(\mathcal{M})$ is positive and it is equal to  $1$ in $D(\mathcal{O})\subset \mathcal{M}$. 
Let $h,g\in C^{\infty}_{\mathrm{c}}(\mathcal{O})$, it holds that
\[
\lim_{k\to\infty} (h,\Delta_{F,k} g) = 0, \qquad \lim_{k\to\infty} (h,\Delta_{+,k} g) = 0
\]
where $\Delta_{F,k}$ is any Feynman propagator of the theory whose action is $I_0+Q_k$ and $\Delta_{+,k}$ the two-point function of the corresponding state. 
\end{lemma}
\begin{proof}
We start analyzing the first limit.
It holds that $\Delta_{F,k}$ is proportional to a fundamental solution of the equation obtained from the differential operator
$P_0-q_k$, namely $(P_0-q_k)\Delta_{F,k} g=ig$.
Hence
\[
-(h,f \Delta_{F,k} g) = \frac{1}{k^2} \left[ {i} (h,g) - (h, P_0\Delta_{F,k} g)   \right].
\]
Since $f$ is $1$ on the support of $h$ we have that $fh=h$ hence using Cauchy-Schwartz inequality  we get
\[
\frac{|(h, f\Delta_{F,k} g)|}{(\|g\|+\|f\Delta_{F,k}g\|) } \leq \frac{1}{k^2}  (\| h\| +\| P_0 h\|)  
\]
where the norm $\|\cdot\|$ is that of $L^2(\mathcal{O}, \dd x)$ and where we use the fact that $P_0$ is formally self-adjoint on $L^2(\mathcal{O}, \dd x)$.
The latter inequality implies that for every $h\in C^{\infty}_{\mathrm{c}}(\mathcal{O})$ 
\[
\frac{(h, f\Delta_{F,k} g)}{(\|g\|+\|f\Delta_{F,k}g\|) }
\]
vanishes in the limit $k\to\infty$ and thus
\[
\frac{f\Delta_{F,k} g}{\|g\|+\|f\Delta_{F,k}g\|}
\]
tends to the $0$ of  $L^2(\mathcal{O}, \dd x)$ under that limit because  $C^{\infty}_{\mathrm{c}}(\mathcal{O})$ are dense.
This implies that 
\[
\lim_{k\to\infty}\|f\Delta_{F,k} g\|  = 0 
\]
We thus have that 
\[
|(h,\Delta_{F,k} g)| \leq \|h\|\|f\Delta_{F,k} g\|
\]
and the right hand side vanishes for $k\to+\infty$.
Hence the first limit we wanted to prove holds. We can use the same strategy to prove the that also the limit of $(h,\Delta_{+,k}g)$ vanishes for large $k$ with the observation that $\Delta_{+,k}$ is a weak solution of $P_0-\q_k$.
\end{proof}

\section{Hadamard expansion of the Minkowski vacuum two-point function and the Wick square}\label{Sec: Hadamard expansion of the Minkowski vacuum two-point function and the Wick square}
 For a massive theory in even-dimensional Minkowski space, the Hadamard distribution is known to depend on an additional arbitrary parameter $\mu$, and it is given by \cite{Brunetti2009}
\[
H^\mu_m(x,y) = \Delta_{S,k}(x,y) + \frac{(-1)^{d/2}}{2 (2\pi)^{d/2}} M^{d/2-1} \log(\frac{\mu^2}{M^2}) \sigma^{\frac{2-d}{4}}  I_{d/2-1}(\sqrt{M^2 \sigma}) \ ,
\]
where $\Delta_{S,k}(x,y)$ is the symmetric contribution of the vacuum two-point function, $I_\nu(x)$ is the modified Bessel function of the first kind and $\sigma = g_{ab} (x-y)^a (x-y)^b$ is the squared geodesic distance. The mass term for the linearized theory is $M^2 = k^2 + m^2_k + \lambda_k \rho$.  In the coincidence limit, $\sigma \to 0$ and 
\[
I_{d/2-1}(x) \simeq \frac{M^{d/2-1} \sigma^{\frac{d-2}{4}}}{2^{d/2-1} \Gamma(d/2)} \ .
\]
Therefore, the $\sigma$ dependence drops and we obtain
\begin{equation}\label{eq:vacuum-Wick}
\omega\left[\frac{\chi^2(x)}{2}\right] = \frac{(-1)^{d/2}}{\Gamma(d/2) (4 \pi)^{d/2}} \bigg (k^2 + m^2_k+\lambda_k \rho\bigg )^{d/2-1} \log(\frac{k^2 + m^2_k + \lambda_k \rho}{\mu^2}) \ .
\end{equation}
We notice that the Hadamard distribution explicitly depend on the scale $k$, which was not the case for the state. The reason is that, under a change in the mass parameter, the massive Minkowski vacuum is mapped into the massive Minkowski vacuum with the rescaled mass; in particular, this implies that the singularity structure of the 2-point function is modified by the mass rescaling.

\printbibliography

\end{document}